\documentclass[12pt,draftclsnofoot,onecolumn,journal]{IEEEtran}
\ifCLASSINFOpdf
\else
\fi
\usepackage{amsmath,dsfont,bbm,epsfig,amssymb,amsfonts,amstext}\usepackage{verbatim,amsopn,cite,subfigure,multirow,multicol,lipsum,xfrac}
\usepackage{amsthm}
\usepackage{mathtools,amsthm}
\usepackage{perpage}
\usepackage{balance}
\usepackage{url}
\usepackage{amsfonts}
\usepackage{epsfig}
\usepackage{caption}
\usepackage{psfrag}	
\usepackage{etoolbox}
\usepackage[Algorithm,ruled]{algorithm}
\usepackage{algorithmicx}
\usepackage{algpseudocode}
\usepackage{pifont}
\usepackage[utf8]{inputenc}
\usepackage[T1]{fontenc}  
\usepackage[nolist]{acronym}
\MakePerPage{footnote}
\usepackage{paralist}
\usepackage{enumitem}
\usepackage{bbm}
\usepackage[process=auto]{pstool}
\usepackage{tikz,pgfplots}
\usetikzlibrary{shapes,arrows}
\newcommand{\setX}{\mathbbmss{X}}
\newcommand{\setH}{\mathbbmss{H}}
\newcommand{\setL}{\mathbbmss{L}}

\newcommand{\setR}{\mathbbmss{R}}

\newcommand{\setS}{\mathbbmss{S}}

\newcommand{\setC}{\mathbbmss{C}}

\DeclareMathOperator*{\argmin}{argmin} 
\DeclareMathOperator*{\argmax}{argmax}

\newcommand{\rmp}{\mathrm{p}}
\newcommand{\rmq}{\mathrm{q}}

\newcommand{\rmR}{\mathrm{R}}

\newcommand{\rmG}{\mathrm{G}}

\newcommand{\her}{\mathsf{H}}

\newcommand{\mae}{\mathcal{E}}

\newcommand{\man}{\mathcal{N}}

\newcommand{\bi}{\mathbf{i}}

\newcommand{\snr}{P/ \sigma^2}

\newcommand{\rmj}{\mathrm{j}}
\newcommand{\rms}{\mathrm{s}}

\newcommand{\RLS}[2]{\mathrm{RLS}_{#1} \left( #2 \right)}
\newcommand{\rls}[2]{\mathrm{rls}_{#1} \left( #2 \right)}

\newcommand{\dec}{\mathrm{dec}}
\newcommand{\reg}{\mathrm{reg}}
\newcommand{\loge}{\left. \mathrm{ln} \right.}
\newcommand{\sign}[1]{ \mathrm{sign} \left( #1 \right)}

\newcommand{\bx}{{\boldsymbol{x}}}

\newcommand{\vv}{\mathrm{v}}

\newcommand{\xx}{\mathrm{x}}
\newcommand{\yy}{\mathrm{y}}
\newcommand{\zz}{\mathrm{z}}

\newcommand{\mac}{\mathcal{C}}

\newcommand{\bv}{{\boldsymbol{v}}}

\newcommand{\bdd}{{\mathbf{d}}}

\newcommand{\dif}{\mathrm{d}}

\newcommand{\by}{{\boldsymbol{y}}}

\newcommand{\bn}{{\boldsymbol{n}}}

\newcommand{\trp}{\mathsf{T}}

\newcommand{\mA}{\mathbf{A}}

\newcommand{\mI}{\mathbf{I}}
\newcommand{\mone}{\mathbf{1}}

\newcommand{\mU}{\mathbf{U}}

\newcommand{\mV}{\mathbf{V}}

\newcommand{\mH}{\mathbf{H}}

\newcommand{\mse}{\mathrm{MSE}}

\newcommand{\asy}{{\mathrm{asy}}}

\newcommand{\Ex}[2]{\mathbbmss{E}_{#2}\left\{#1\right\}}
\newcommand{\norm}[1]{\lVert #1 \rVert}
\newcommand{\set}[1]{\left\lbrace #1 \right\rbrace}
\newcommand{\brc}[1]{\left( #1 \right)}
\newcommand{\dbc}[1]{\left[ #1 \right]}
\newcommand{\real}[1]{\mathbbm{Re}\left\lbrace #1 \right\rbrace}
\newcommand{\img}[1]{\mathbbm{Im}\left\lbrace #1 \right\rbrace}

\newcommand{\abs}[1]{\lvert #1 \rvert}

\newcommand{\Sp}[1]{\mathrm{Supp}\left( #1 \right) }

\newtheoremstyle{mystyle}%                % Name
  {}%                                     % Space above
  {}%                                     % Space below
  {\it}%                                     % Body font
  {}%                                     % Indent amount
  {\bfseries}%                            % Theorem head font
  {:}%                                     % Punctuation after theorem head
  { }%                                    % Space after theorem head, ' ', or \newline
  {}%                                     % Theorem head spec (can be left empty, meaning `normal')

\theoremstyle{mystyle}

\newtheorem{definition}{Definition}

\newtheorem{theorem}{Theorem}
\newtheorem{remark}{Remark}

\newtheorem*{tuning}{Tuning Approach}
\newtheorem{example}{Example}
\algnewcommand\algorithmicLet{\textbf{Let}}
\algnewcommand\Let{\item[\algorithmicLet]}
\algnewcommand\algorithmicSet{\textbf{Set}}
\algnewcommand\Set{\item[\algorithmicSet]}

\algnewcommand\algorithmicInitiate{\textbf{Initiate}}
\algnewcommand\Initiate{\item[\algorithmicInitiate]}
\algnewcommand\algorithmicStart{\textbf{Begin}}
\algnewcommand\Begin{\item[\algorithmicStart]}
\algnewcommand\algorithmicEnd{\textbf{End}}
\algnewcommand\End{\item[\algorithmicEnd]}

\algnewcommand\algorithmicOutP{\textbf{Output:}}
\algnewcommand\Out{\item[\algorithmicOutP]}

\newcounter{bar}

\begin{document}

\begin{acronym}
\acro{mimo}[MIMO]{multiple-input multiple-output}
\acro{ssk}[SSK]{space shift keying}
\acro{bpsk}[BPSK]{binary PSK}
\acro{csi}[CSI]{channel state information}
\acro{awgn}[AWGN]{additive white Gaussian noise}
\acro{iid}[i.i.d.]{independent and identically distributed}
\acro{ut}[UT]{user terminal}
\acro{bs}[BS]{base station}
\acro{sm}[SM]{spatial modulation}
\acro{masm}[MA-SM]{multiple-active SM}
\acro{glse}[GLSE]{generalized least squared error}
\acro{rls}[RLS]{regularized least-squares}
\acro{rhs}[r.h.s.]{right hand side}
\acro{lhs}[l.h.s.]{left hand side}
\acro{wrt}[w.r.t.]{with respect to}
\acro{rs}[RS]{replica symmetry}
\acro{rsb}[RSB]{replica symmetry breaking}
\acro{papr}[PAPR]{peak-to-average power ratio}
\acro{rzf}[RZF]{regularized zero forcing}
\acro{psk}[PSK]{phase shift keying}
\acro{pam}[PAM]{pulse shape modulation}
\acro{qam}[4-QAM]{quadrature amplitude modulation}
\acro{snr}[SNR]{signal-to-noise ratio}
\acro{rf}[RF]{radio frequency}
\acro{tdd}[TDD]{time division duplexing}
\acro{mf}[MF]{match filtering}
\acro{gamp}[GAMP]{generalized approximate message passing}
\acro{map}[MAP]{maximum-a-posteriori-probability}
\acro{mse}[MSE]{mean square error}
\acro{mmse}[MMSE]{minimum MSE}
\acro{ml}[ML]{maximum-likelihood}
\acro{as}[AS]{antenna selection}
\acro{svd}[SVD]{singular value decomposition}
\acro{lasso}[LASSO]{least absolute shrinkage and selection operator}
\acro{aep}[AEP]{asymptotic equipartition property}
\end{acronym}

\title{Detection of Spatially Modulated Signals via RLS: Theoretical Bounds and Applications}
% Authors
\author{
\IEEEauthorblockN{
Ali Bereyhi, \textit{Member, IEEE},
Saba Asaad, \textit{Member, IEEE},
Bernhard G\"ade,\\
Ralf R. M\"uller, \textit{Senior Member, IEEE}, and
H. Vincent Poor, \textit{Fellow, IEEE}
}
\thanks{This work has been presented in parts at the 2019 IEEE International Symposium on Information Theory (ISIT) in Paris \cite{bereyhi2019rls}.}
\thanks{Ali Bereyhi, Saba Asaad, Bernhard G\"ade and Ralf R. M\"uller are with the Institute for Digital Communications, Friedrich-Alexander Universit\"at Erlangen-N\"urnberg, Germany, \textit{\{ali.bereyhi, saba.asaad, bernhard.gaede, ralf.r.mueller\}@fau.de}. H.~Vincent~Poor is with the Electrical Engineering Department, Princeton University, NJ 08544, \textit{poor@princeton.edu}.}
\thanks{This work was supported by Deutsche Forschungsgemeinschaft under Project-No. MU 3735/7-1, and by the U.S. National Science Foundation under Grant CCF-1908308.}
}


\IEEEoverridecommandlockouts

\maketitle

\begin{abstract}
This paper characterizes the performance of massive multiuser spatial modulation MIMO systems, when a regularized form of the least-squares method is used for detection. For a generic distortion function and right unitarily invariant channel matrices, the per-antenna transmit rate and the asymptotic distortion achieved by this class of detectors is derived. Invoking an asymptotic characterization, we address two particular applications. Namely, we derive the error rate achieved by the computationally-intractable optimal Bayesian detector, and we propose an efficient approach to tune a LASSO-type detector. We further validate our derivations through various numerical experiments.

\end{abstract}

\begin{IEEEkeywords}
	Multiple-active spatial modulation, box-LASSO, regularized least-squares, massive MIMO, maximum-a-posteriori-probability detection
\end{IEEEkeywords}

\IEEEpeerreviewmaketitle

\section{Introduction}
\label{sec:intro}

\Ac{sm} is a multiple-antenna transmission technique in which information is conveyed not only through transmitted symbols, but also via the indices of the transmit antennas \cite{mesleh2006spatial,mesleh2008spatial,jeganathan2008spatial,di2011spatial,yang2014design}. Initial realizations of \ac{sm} mainly performed index modulation\footnote{By index modulation, we mean that the data is only conveyed via the index of the selected antennas.} and are often referred to as \ac{ssk} techniques in the literature. Later developments extended the idea to more generalized schemes, some of which can be followed in \cite{jeganathan2008generalized,jeganathan2009space,di2010general,Mseleh2011SSK,di2011bit,di2012space,popoola2013error,basar2011space,bian2015differential} and the references therein. Among the various available schemes, \ac{masm} is the most generic form in which the transmitter performs index modulation over a subset of multiple transmit antennas. The active antennas further transmit information using a generic constellation set \cite{fu2010generalised,wang2012generalised,cheng2015enhanced}.

This paper investigates the large-system performance of a classic \ac{masm} system, when the receiver detects the transmitted data symbols jointly via a \ac{rls}-based algorithm. The motivation behind such a study is demonstrated in the shadow of two main facts:
\begin{inparaenum}
	\item[(i)] Compatibility of \ac{masm} with massive \ac{mimo} systems.
	\item[(ii)] Generality of \ac{rls}-based detection.
\end{inparaenum}
In the sequel, we briefly discuss these motivational facts.

\subsection{Massive SM MIMO Systems}
Theoretical analyses and implementational validations indicate that massive \ac{mimo} technology will be a key element in future generations of cellular networks \cite{hoydis2013massive,larsson2014massive,bjornson2015massive,bjornson2016massive,gao2018low,kuehne2018analog,kuehne2020performance}. Efforts for enabling this technology as a cost-efficient standard that can be commercialized have led to several paradigm-shifts in \ac{mimo} designs, among which hybrid architectures \cite{alkhateeb2014channel,gao2016energy,mendez2016hybrid,asaad2017asymptotic,bereyhi2019papr}, nonlinear precoding schemes \cite{bereyhi2017wsa,sedaghat2017least,bereyhi2019glse,bereyhi2019rls}, user and antenna selection techniques \cite{li2014energy,gao2015massive,liu2016efficient,asaad2017tas,bereyhi2017isit,asaad2018massive,asaad2018optimal,bereyhi2018stepwise}, and single \ac{rf} chain \ac{mimo} \cite{sedaghat2014novel,sedaghat2016load,gade2017novel,bereyhi2020single} are prominent examples.

From implementational viewpoints, \ac{sm} addresses the same issue as antenna selection: \textit{Reducing the overall \ac{rf} cost} by using fewer \ac{rf} chains than the number of transmit antennas. The key difference here is that \ac{sm} increases the data rate, compared to the antenna selection technique, by transmitting extra information bits via \textit{index modulation}. This enhancement is achieved at the expense of loosing some diversity gain at the transmitter side. This trade-off was investigated in \cite{gade2019fair}. The results show that \ac{sm} is a suitable technique when it is employed in multiuser massive \ac{mimo} systems with sufficiently many receive antennas\footnote{See the concluding points of \cite{gade2019fair}.}.

The findings of \cite{gade2019fair} anticipate that \ac{sm} will, in practice, be deployed for uplink transmission in massive \ac{mimo} systems. This motivates us to characterize the performance of \ac{sm} in its generic form, i.e., \ac{masm}, in the large-system limit.

\subsection{RLS-based Detection}
Detection of a spatially modulated signal mainly deals with two tasks: \textit{support recovery}, and \textit{symbol detection}. In fact, given noisy and linearly projected observations obtained through the channel, the receiver needs to detect both the indices of active antennas and their corresponding transmitted symbols. The former is used to recover the information conveyed via index modulation, while the latter recovers the information transmitted via conventional modulation.

In the Bayesian framework, the optimal approach for detection is to find the realization of the transmit signals whose posterior probability conditioned on the received signal is maximized, i.e., \ac{map} detection \cite{jeganathan2008spatial,wen2015low}. However, due to the \textit{sparse} nature of spatially modulated signals, optimal detection reduces to an \textit{integer programming} problem whose complexity grows exponentially with signal dimension, and hence, is not tractable in practice. An alternative approach is to look at the detection task as a \textit{sparse recovery} problem\footnote{Despite mathematical equivalency of the two tasks, a conventional sparse recovery problem differs from \ac{sm} detection in the fact that the latter is not \textit{necessarily} underdetermined.} in which a sparse signal is to be recovered from a set of noisy linear observations \cite{candes2006robust,donoho2006compressed}. Following this alternative viewpoint, several lines of work have proposed low-complexity detection algorithms using sparse recovery techniques; see for example \cite{yu2012compressed,garcia2015low,xiao2017compressed,hemadeh2018compressed}.

For Gaussian channels, both Bayesian and sparse recovery techniques are formulated similarly: The detector minimizes the \textit{residual sum of squares}, i.e., the error between the received signal and the projection of possible transmit signals over the given channel, subject to some side constraints restricting the sparsity and constellation. In the context of linear regression, this approach is called the method of \ac{rls}. Most \ac{sm} detection algorithms are mathematically equivalent to \ac{rls}. Hence, analyzing a generic \ac{rls}-based detection scheme  characterizes of a large scope of algorithms. Following this, we consider a generic \ac{rls}-based detector and investigate its performance in various respects. The generality of this setting enables us to address multiple special cases, among which we discuss the optimal Bayesian detector and the box-constrained \ac{lasso} detector in details.

\subsection{Contributions and Organization}
This work characterizes the asymptotic performance of a multiuser \ac{mimo} setting in which users employ the \ac{masm} for uplink transmission and the receiver uses an \ac{rls}-based detector to jointly detect the transmitted signals. There are some earlier works that investigate special cases of this setting; e.g., \cite{atitallah2017box,thrampoulidis2018symbol,alrashdi2019precise,alrashdi2020box}. Nevertheless, the analyses and results of this paper are new in various respects:
\begin{itemize}
	\item The standard approach to the analysis of massive \ac{sm} \ac{mimo} systems is to model the transmit signal as an \ac{iid} sparse sequence; see for example \cite{atitallah2017box}. Although this assumption simplifies large-system analyses and can lead to a good approximation, it does not precisely model an \ac{sm} system. In fact, using index modulation via a pre-defined codebook, the transmit signals \textit{are not necessarily \ac{iid}}, even if the codebook is generated randomly; see Example~\ref{ex:1} in Section~\ref{sec:SM}. We address this issue by deviating from the classic approach and considering a precise model for \ac{sm} signals. Invoking our prior results on asymmetric \ac{map} estimation \cite{bereyhi2018MAP}, we characterize the asymptotic performance by considering this more accurate model.
	\item In contrast to earlier works in the literature, we study the \ac{sm} system by considering a generic model for the channel matrix. Namely, we consider the channel matrix to be a \textit{complex-valued right unitarily invariant} random matrix. This includes various well-known models for the fading process, e.g., the standard Rayleigh fading model.
	\item The analysis in this work considers a generic form of \ac{rls}-based detection. This not only allows us to derive error bounds for optimal and sub-optimal detection algorithms, but also to address various design tasks. As an example, we discuss the particular application of tuning box-constrained \ac{lasso} detectors in this paper.
\end{itemize}

The remaining parts of this manuscript are organized as follows: The setting is modeled in Section~\ref{sec:sys}. The main results and derivations are presented in Section~\ref{sec:Large}. As an application of the results, the optimal Bayesian detection algorithm is characterized in Section~\ref{sec:App1}. Applications of the results to the analysis and tuning of box-constrained \ac{lasso} detectors are discussed in Section~\ref{sec:App2}. The paper is finally concluded in Section~\ref{sec:Conc}.

\subsection{Notation and Basic Definitions}
Scalars, vectors and matrices are represented with non-bold, bold lower-case and bold upper-case letters, respectively. $\mH^{\her}$ indicates the transposed conjugate of $\mH$, and $\mI_N$ is an $N\times N$ identity matrix. The $\ell_p$-norm of $\bx$ is denoted by $\norm{\bx}_p$ and $\norm{\bx}_0$ is the ``$\ell_0$-norm'' of $\bx$ which counts the number of non-zero entries. $\log$ and $\loge \left(\cdot\right)$ indicate the binary and natural logarithm, respectively. $\Ex{}{x}$ is expectation with respect to $x$. The binary entropy function $H_2\brc{p}$ is defined as
\begin{align}
	H_2\brc{p} =
	-p\log\brc{p}-\brc{1-p}\log\brc{1-p},
\end{align}
for some $p\in\brc{0,1}$ and is zero if $p\in\set{0,1}$. $\lfloor x \rfloor$ denotes the floor function that rounds $x$ to the nearest integer $i \leq x$. $\setR$ and $\setC$ refer to the real axis and the complex plane, respectively. For $z\in\setC$, $z^*$, $\real{z}$ and $\img{z}$ denote the complex conjugate, real part and imaginary part of $z$, respectively. $\mathcal{CN}\brc{\eta,\sigma^2}$ represents the complex Gaussian distribution with mean $\eta$ and variance $\sigma^2$. For a given set $\setS$, we use the notation $\setS_0$ to denote $\setS_0 \coloneqq \set{0}\cup\setS$. $\setX^{\rm C}$ denotes the complement of $\setX$ with respect to a mother set. For sake of brevity, $\set{1,\ldots,N}$ is abbreviated as $\dbc{N}$.

For sake of brevity, we further define the Stieltjes transform and the $\rmR$-transform for large random matrices as follow: Consider matrix $\mA\in\setC^{N\times M}$, and let $\vartheta_1, \ldots, \vartheta_M$ denote the eigenvalues of $\mA^\her \mA$, i.e., the singular values of $\mA$. Define the density of the states as
\begin{align}
	\rmp_\mA^{\brc{M}} \brc{\vartheta} = \frac{1}{M} \sum_{m=1}^M \mone\set{ \vartheta = \vartheta_m} 
\end{align}
and let $\rmp_\mA\brc{\vartheta} $ be its limit when $M$ and $N$ go to infinity with a fixed ratio, i.e, $N/M$ is fixed. We refer to $\rmp_\mA\brc{\vartheta} $ as the \textit{asymptotic singular value distribution} $\mA$. For this distribution, the Stieltjes transform is defined as
\begin{align}
	\rmG_\mA \brc{s} = \int \frac{\rmp_\mA \brc{\vartheta} }{\vartheta - s} \dif \vartheta
\end{align}
for some complex $s$ with $\img{s} \geq 0$. Denoting the inverse of $\rmG_\mA\brc{\cdot}$ with respect to composition with $\rmG_\mA^{-1}\brc{\cdot}$, the $\rmR$-transform is then defined as %
%\begin{align}
$\rmR_\mA\brc{\omega} = \rmG_\mA^{-1}\brc{-\omega} - {\omega}^{-1}$, 
%\end{align}
such that
\begin{align}
	\lim_{\omega \to 0} \rmR_\mA\brc{\omega} = \int \vartheta \rmp_\mA \brc{\vartheta}  \dif \vartheta.
\end{align}
We use these transforms to represent the main results.

\section{Problem Formulation}
\label{sec:sys}
We consider a Gaussian \ac{mimo} broadcast channel, in which $K$ users transmit uplink signals to a single \ac{bs}. Each user is equipped with $M_{\rm u}$ antennas and $L_{\rm u}$ transmit \ac{rf} chains. This means in each transmission time interval, only $L_{\rm u}$ transmit antennas are active at each user terminal. We denote the fraction of active antennas by $\eta = L_{\rm u} / M_{\rm u}$ and refer to it as the \textit{activity ratio}. The \ac{bs} is equipped with $N$ receive antennas. The uplink channel in this case is compactly represented by
\begin{align}
\by = \left. \mH \right. \bx + \bn
\end{align}
where $\mH$, $\bx$ and $\bn$ represent the channel matrix, transmit signal and \ac{awgn}, respectively, and fulfill the following constraints:
\begin{enumerate}[label=(\alph*)]
\item $\mH\in\setC^{N \times M}$ with $M = K M_{\rm u}$ being the total number of available transmit antennas in the network. The entries of $\mH$ represent channel gains between transmit and receive antennas over a single time-frequency resource.

The channel is assumed to experience quasi-static fading with slow time variations, meaning that the gains are fixed within the given frequency band during a \textit{coherence time interval} which is considerably larger than a symbol interval. We consider a generic stochastic model for the fading process. Namely, it is assumed that $\mH$ is a \textit{right unitarily invariant} random matrix. This means that $\mH$ has a \ac{svd}
\begin{align}
\mH = \mU \mathbf{\Sigma} \mV^\her, \label{svd}
\end{align}
where $\mathbf{\Sigma}\in \setR^{N\times M}$ contains the singular values of $\mH$ on the main diagonal and zeros elsewhere, $\mU\in\setC^{N\times N}$ is a unitary matrix, and $\mV\in\setC^{M\times M}$ is a \textit{Haar-distributed} unitary matrix, i.e., $\mV$ is distributed uniformly over the set of unitary matrices.

The ensemble of right unitarily invariant random matrices includes a variety of fading models including the standard \ac{iid} Rayleigh fading model.
\item Transmit signal $\bx\in\setC^{M}$ is given by
\begin{align}
\bx = \left[ \bx_1^\trp,\ldots,\bx_K^\trp \right]^\trp
\end{align}
where $\bx_k\in\setC^{M_{\rm u}}$ represents the transmit signal of user $k$ and is constructed by mapping the information symbols of the user to symbols from the constellation set.
\item $\bn\in \setC^{N}$ is a complex-valued \ac{iid} random vector whose entries are Gaussian with zero-mean and variance $\sigma^2$, i.e. $\bn \sim\mathcal{CN}\brc{\boldsymbol{0}, \sigma^2 \mI_N}$.
\item $\by\in\setC^{N}$ denotes the receive signal whose entries are in general mutually coupled.
\end{enumerate}

It is assumed that the system operates in the \ac{tdd} mode, which is typical for massive \ac{mimo} systems. The uplink channels are estimated prior to data transmission by sending $K$ orthogonal pilot sequences within the \textit{training phase}. To keep the analysis tractable, we neglect the impact of estimation errors and assume that the \ac{csi} is perfectly available at the \ac{bs}.

\subsection{Spatial Modulation}
\label{sec:SM}
To construct transmit signals $\bx_1,\ldots,\bx_K$ from the information symbols, the users employ \ac{masm}: Each user selects a subset of $L_{\rm u}$ transmit antennas and sends $L_{\rm u}$ independent modulated symbols over them using a standard modulation scheme, e.g. \ac{psk}. The key difference to the conventional transmission is that the index of the selected subset is further specified by data symbols, and hence it carries information.

To illustrate the modulation scheme precisely, let $\bdd_k$ denote the sequence of information symbols being sent by user $k$. Assume $\bdd_k$ is an \ac{iid}\footnote{We assume that the information symbols are interleaved, such that the temporal correlation is negligible.} binary sequence with uniform distribution. The \ac{rf} chains at the user terminals are assumed to transmit complex signals whose constellation points are taken from $\setS$. We assume that $\setS$ contains $2^S$ distinct points, which is the case in practice, e.g. in \ac{psk}. Given the information sequence $\bdd_k$, user $k$ constructs $x_k$ as follows:
\begin{enumerate}
\item \textit{Codebook generation:} Consider all possible tuples of $L_{\rm u}$ antennas selected out of $M_{\rm u}$ available antennas at each user. A subset of $2^I$ distinct tuples is selected, where
\begin{align}
I = \left\lfloor \log {M_{\rm u}\choose L_{\rm u}} \right\rfloor.
\end{align}
To each of these subsets a \textit{modulation index} is assigned. We refer to this indexed subset as the \textit{codebook}.
\item \textit{Index modulation:} For given sequence $\bdd_k$, user $k$ chooses modulation index $i_k\in [ 2^I ]$ from the codebook, such that the first $I$ bits of $\bdd_k$ are the binary representation of $i_k$. 
\item \textit{Modulating multiple streams:} Over the active antennas selected by index $i_k$, user $k$ transmits $s_{k}\brc{m} \in \setS$, with $m\in\setL\brc{i_k}$, where $\setL \brc{i_k} \subseteq [M_{\rm u}]$ denotes the subset of $L_{\rm u}$ antennas which correspond to modulation index $i_k$ and is referred to as the \textit{active support} of user $k$.
\end{enumerate}
The $m$-th transmit entry of user $k$, i.e. $x_{k,m}$ for $m\in [M_{\rm u}]$, is therefore written as
\begin{align}
x_{k,m}=
\begin{cases}
s_{k}\brc{m} & m\in \setL\brc{i_k}\\
0 & m\notin \setL\brc{i_k}
\end{cases}. \label{eq:x_km}
\end{align}

From \eqref{eq:x_km}, it is concluded that $\bx_k$ is an $L_{\rm u}$-sparse vector, i.e. only $L_{\rm u}$ entries are non-zero. The transmit signal $\bx$ is hence an $L$-sparse vector, where $L=K L_{\rm u}$. This means that the \textit{sparsity factor} of the transmit signal is
\begin{align}
\eta = \frac{\norm{\bx}_0}{M} = \frac{L}{M} = \frac{L_{\rm u}}{M_{\rm u}}.
\end{align}

It is worth to indicate that the transmit signal is in general not \ac{iid} distributed. In fact, by index modulation, the entries of $\bx_k$ become statistically dependent. Furthermore, following the asymmetry imposed by codebook generation, the marginal distributions of transmit entries are not necessarily identical. This point is clarified through the following toy-example.

\begin{example}
	\label{ex:1}
Consider a scenario with a single user equipped with $M_{\rm u} = 5$ antennas and $L_{\rm u} = 2$ \ac{rf} chains. We assume \ac{bpsk} transmission over the active antennas, i.e. $S=1$. For this setting, $I= 3$ which means that the codebook consists of $8$ distinct pairs of antennas each indexed by a modulation index from $0$ to $7$. Let the codebook be 
\begin{align}
\mac = \set{ \brc{j,\ell}: \  j\in\set{1,2} \;  \mathrm{and} \;  j < \ell \leq 5  } \cup \set{\brc{3,4}}
\end{align}
with some indexing. The user in this case maps a binary symbol sequence of length $I+L_{\rm u} = 5$ into a transmit signal $\bx$ where the first $I=3$ symbols specify the index of active antennas, and the remaining $L_{\rm u} = 2$ bits are transmitted via the selected antennas using the \ac{bpsk} constellation. Let us define conditional distribution $ \rmq_m \brc{ x \vert \bx_0 } $ for $m \in \dbc{5}$ as
\begin{align}
\rmq_m \brc{ x \vert \bx_0 } \coloneqq \Pr\set{x_m = x \left\vert \bx_{\backslash m} = \bx_0 \right. }
\end{align}
where $x_m$ is entry $m$ of $\bx$, and $\bx_{\backslash m}$ denotes a four-dimensional vector constructed from $\bx$ by dropping $x_m$, respectively. 

For this setting, it is straightforward to write
\begin{align}
\rmq_m \brc{ x \vert \bx_0 } = 
\begin{cases}
\mone\set{ x = 0 } & \norm{\bx_0}_0 = 2\\
\dfrac{1}{2} \mone\set{ x = \pm 1 } & \norm{\bx_0}_0 \neq 2\\
\end{cases}
\end{align}
which indicates that the entries of $\bx$ are mutually dependent. Given the codebook, we can write
%\begin{subequations}
\begin{align}
\rmp_{x_1} \brc{ x } = 
\begin{cases}
0.5 & x=0\\
0.25 & x=\pm 1\\
\end{cases}
\qquad \mathrm{and} \qquad
\rmp_{x_5} \brc{ x } = 
\begin{cases}
0.75 & x=0\\
0.125 & x=\pm 1\\
\end{cases}
\end{align}
%\end{subequations}
where $\rmp_{x_m} \brc{ x }$ denotes the marginal distribution of $x_m$. This observation indicates that the entries of $\bx$ are not identically distributed\footnote{The necessary condition for identical marginal distributions is that the total number of tuples be an integer power of 2. This is however not the case for many choices of $M_{\rm u}$  and $L_{\rm u}$.}.
\end{example}

\subsection{RLS-based Detection Algorithms}
The \ac{bs} detects the transmit signal from the received signal by employing a generic \ac{rls}-based detection algorithm. This algorithm first determines a soft estimation of the transmit signal, for the given \ac{csi}, using the \ac{rls} recovery
\begin{align}
\RLS{\setX}{\by \vert \mH} \coloneqq \argmin_{\bv\in\setX_0^M} \left. \norm{\by-\mH \bv}^2 + f_\reg \brc{\bv} \right. . \label{eq:RLS_vector}
\end{align}
Here, $\setX$ is a superset\footnote{$\setX$ is usually set to the convex hull of $\setS$.} of the constellation set $\setS$, i.e., $\setS\subseteq \setX$, and $f_\reg \brc\cdot$ is some regularization function. The soft estimation is hence determined as 
\begin{align}
\bx^\star = \RLS{\setX}{\by \vert \mH}
\end{align}
The detected signal is then given by mapping the soft estimation to a vector in $\setS_0^M$. This means
\begin{align}
\hat{\bx} = f_\dec \brc{\bx^\star}
\end{align}
where $f_{\rm dec} \brc{\cdot}: \setX_0^M \mapsto \setS_0^M$ is a decision function, e.g. the sign function.
\subsection{Special Forms of RLS-Detectors}
The considered recovery scheme includes a large scope of detection algorithms. From the Bayesian points of view, the algorithm reduces to the \textit{optimal} detector, i.e. \ac{map} detector, by setting $f_{\rm dec} \brc{\cdot}$ to be the identity function, i.e. $\hat{\bx} = \bx^\star$,  $\setX=\setS$ and 
\begin{align}
f_\reg \brc{\bv} = - \sigma^2 \loge \rmp_{\bx} \brc{\bv} \label{eq:prior}
\end{align}
with $\rmp_\bx$ denoting the prior distribution of $\bx$ imposed by the statistical model of data. By changing $f_\reg \brc\cdot$, while keeping $\setX=\setS$ and $f_{\rm dec} \brc{\cdot}$ to be the identity function, the algorithm is interpreted as a \textit{mismatched} \ac{map} detector. %in the Bayesian framework.

For $\setX=\setS$ and large $M$, the algorithm is computationally intractable, and hence, it is infeasible to implement it, in practice. As a result, convex forms of the algorithm are often used in practice. A well-known example of such detectors is the \textit{box-constrained \ac{lasso}}. In the box-constrained \ac{lasso}, $\setX$ is set to a convex set which contains the constellation points. The regularization term is set proportional to the $\ell_1$-norm of $\bx$ which approximates the sparse prior of the transmit signal with a Laplace distribution. Using the \ac{rls} recovery algorithm, the soft estimation of the  transmit signal contains entries which are either zero or a complex number in $\setX$. For decision, $f_{\rm dec} \brc{\cdot}$ is set to be an entry-wise hard thresholding operator which maps the points within a certain decision region to a corresponding constellation point. An example of the box-constrained \ac{lasso} is given below:
\begin{example}
Assuming a binary \ac{psk} transmission, the constellation set is $\setS = \{-\sqrt{P}, +\sqrt{P}\}$, for some positive real $P$. In this case, a possible choice for $\setX$ is $\setX = [-B,B]$ for some $B \geq \sqrt{P}$. The soft estimation is then given by
\begin{align}
\bx^\star \coloneqq \argmin_{\bv\in [-B,B]^M} \left. \norm{\by-\mH \bv}^2 + \lambda \norm{\bv}_1 \right. .
\end{align}
The detected signal is given by setting the $M-L$ entries of $\bx^\star$ with smallest absolute values zero and detecting the rest as $\hat{x}_m = \sqrt{P} \left. \sign{x^\star_m} \right.$. An equivalent representation of the decision function in this case is 
\begin{align}
f_{\rm dec} \brc{ x } = \sqrt{P} \left. \sign{x} \right. \mone\set{\abs{x} > \epsilon}
\end{align}
for some $\epsilon \in [0,B]$, where $\mone\set{\cdot}$ is the indicator function.
\end{example}
%This generic form covers a large special cases. For instance,
%\begin{itemize}
%\item For $f_\reg \brc\bv \propto \norm{\bv}_0$ and $\setX=\setS$, the detector recovers optimal \ac{rls} recovery scheme.
%\item For $f_\reg \brc\bv \propto \norm{\bv}_1$ and $\setX=\setS$, the detector reduces to \ac{lasso} recovery over the exact signal constellation.
%\item Setting $f_\reg \brc\bv \propto \norm{\bv}_1$, and $\setX$ to be a convex subset of the complex plane, containing $\setS$, the detector reduces to \ac{lasso} recovery with box-relaxation which is often referred to as \textit{box-\ac{lasso}} detector for \ac{gsm}.
%\end{itemize}
\subsection{Performance Measures}
Using spatial modulation, the data rate is increased by $I$ bits per transmission compared to the conventional antenna selection approach \cite{asaad2018massive}. In fact, the data rate per user in this case is
\begin{align}
R_{\rm u} = I + L_{\rm u} S = I + \eta M_{\rm u} S.
\end{align}
This increase is obtained at the expense of reducing the diversity gain. To characterize~the~performance of this transmission technique over the noisy channel, a distortion metric is further considered. The common metric is the average error rate which is defined as the probability of bit flips averaged over the block size, i.e. $M$. Nevertheless, for a general case, the distortion can be determined with respect to a generic distortion measure.

We characterize the performance by defining two metrics, namely the \textit{per-antenna transmit rate} and \textit{  average distortion}, as following:
\begin{definition}[Per-antenna transmit rate]
For each user in the network, the per-antenna transmit rate is defined as
\begin{align}
\bar{R}   \coloneqq \frac{R_{\rm u}}{M_{\rm u}}.
\end{align}
\end{definition}

$\bar{R}$ determines the number of bits achieved per transmit antenna, regardless of whether it is active or passive. To characterize the distortion, we consider a generic metric which includes conventional distortion metrics. 

\begin{definition}[Average distortion]
\label{dist}
Consider the soft estimation $\bx^\star$. The   average distortion is defined as
\begin{align}
D \brc{M} \coloneqq \frac{1}{M} \sum_{m=1}^M \Ex{ F_{\rm D} \brc{x^\star_m;x_m} }{} \label{eq:dist}
\end{align}
for some distortion function $F_{\rm D} \brc{\cdot;\cdot}: \setX_0 \times \setS_0 \mapsto \setR$. 
\end{definition}

Definition~\ref{dist} reduces to several distortion metrics including the conventional error probability and the \ac{mse}. In fact, by setting 
\begin{align}
F_{\rm D} \brc{x^\star_m;x_m} = \mone \set{ f_\dec \brc{x^\star_m} \neq x_m }, \label{eq:error_func}
\end{align}
$D\brc{M}$ calculates the average error rate. An alternative is
\begin{align}
F_{\rm D} \brc{x^\star_m;x_m} = \abs{x^\star_m - x_m}^2, \label{eq:MSE_Dist}
\end{align}
which determines the average \ac{mse} of the soft estimation.

\subsection{Asymptotic Analysis}
We study the system performance in the asymptotic regime. To this end, we assume a sequence of settings indexed by $N$, such that the number of transmit antennas per receive antenna is fixed. This means that $N$ grows large in this sequence, while
\begin{align}
\xi = \frac{M}{N}
\end{align}
is kept fixed. We refer to $\xi$ as the \textit{effective load}. For this sequence of settings, a corresponding sequence of performance metrics, with respect to a given measure of performance, is derived. Each entry of this new sequence gives the performance metric of its corresponding setting in the former sequence. The asymptotic performance is then characterized by the limit of this sequence when $N$ tends to infinity.

In general, the asymptotic performance is interpreted in two different ways:
\begin{enumerate}
\item It describes a case in which a massive number of ordinary multi-antenna users, with small antenna arrays, transmit uplink signals to a \ac{bs} with a large antenna array. In this case, the effective load is
\begin{align}
\xi = \frac{K M_{\rm u}}{N} = \left. \alpha \right. M_{\rm u}.
\end{align}
where $\alpha \coloneqq K/N$ is the \textit{system load}.
\item An alternative interpretation is given by considering a scenario in which few sophisticated terminals communicate with a \ac{bs}. In this case, the system load converges asymptotically to zero, and $M_{\rm u} \gg K$. The effective load is hence given by
\begin{align}
\xi = \frac{K M_{\rm u}}{N} = K \xi_{\rm u},
\end{align}
where $\xi_{\rm u} \coloneqq {M_{\rm u}}/{N}$ is the \textit{per-user load}.
\end{enumerate}
The results of this work address both interpretations. Nevertheless, whenever needed, we refer to these interpretations as \textit{massive user case and massive array case}, respectively.

\section{Large-system Characterization}
\label{sec:Large}
In this section, we derive the asymptotic limits of the performance metrics, i.e., the per-antenna transmit rate and the average distortion. For the per-antenna rate, the exact value of the metric is calculated for any dimension. The study of this metric aims to characterize the rate growth in terms of the transmit array size, asymptotically. We follow the derivation in this case by a standard algebraic approach.

In contrast to the per-antenna rate, the average distortion is not necessarily determinable in a tractable way. In this respect, the large-system analysis of this metric intends to calculate the limiting value of the distortion using some advanced analytical tools. To this end, we invoke the replica method which has been developed in the context of statistical mechanics, and is accepted as an analytical tool in information theory and signal processing.

\subsection{Per-antenna transmit rate}
Theorem~\ref{th1} describes variations of the per-antenna transmit rate with respect to the transmit array size at user terminals, i.e. $M_{\rm u}$. This result does not consider the large-system limit, and is valid for any dimension.
\begin{theorem}
\label{th1}
For any transmit array size $M_{\rm u}$, there exists a constant $C\in \left( C_{\rm d} , C_{\rm u} \right]$, such that the per-antenna transmit rate is given by
\begin{align}
\bar{R} =  \left. \eta \right.  S + H_2\brc{\eta} + \frac{C - \log M_{\rm u} }{2M_{\rm u}}  \label{eq:Theorem1}
\end{align}
where
\begin{subequations}
\begin{align}
C_{\rm d} &= \log \frac{\pi}{2e^4}- \log \brc{\eta-\eta^2}\\
C_{\rm u} &= \log \frac{e^2}{4\pi^2}- \log \brc{\eta-\eta^2}.
\end{align}
\end{subequations}
\end{theorem}
\begin{proof}
Using Stirling's formula, it is shown in Appendix~\ref{app:A} that $I$ is bounded as
\begin{align}
I_{\rm d} < I \leq I_{\rm u},
\end{align}
where $I_{i}$ for $i \in \set{\mathrm{d},\mathrm{u}}$ are given by
\begin{subequations}
	\begin{align}
	I_{i} &\coloneqq \frac{C_{i} - \log M_{\rm u}}{2 }  + M_{\rm u} H_2 \brc{\eta}.
	\end{align}
\end{subequations}
Invoking these upper and lower bounds, the result in \eqref{eq:Theorem1} is concluded straightforwardly using the method of intervals. The detailed derivations are given in Appendix~\ref{app:A}.
\end{proof}

Theorem~\ref{th1} depicts that the per-antenna transmit rate tends to $\left. \eta \right.  S + H_2\brc{\eta}$, as the number of transmit antennas at user terminals grows large. Hence, the gain achieved via index modulation, compared to simply applying antenna selection, is approximately $H_2\brc{\eta}$. For the massive array case, this gives an accurate approximation, while in the massive user case, it requires further modification via the residual terms.

Theorem~\ref{th1} characterizes the \textit{rate loss} caused by deactivating antennas at the user terminals \cite{cakmak2018capacity}. In fact, by considering the case with full transmit complexity as the reference, one observes \ac{masm} reduces the rate loss by $H_2\brc{\eta}$, compared to the antenna selection technique.

\subsection{Average distortion}
%For finite dimensions, the computational complexity of calculating the   average distortion for several choices of recovery support $\setX$ and regularization term $f_\reg\brc\cdot$ grow exponentially large. For several computationally tractable choices, moreover, the distortion is not derived in a closed form via basic analytic tools. Alternatively, we can focus on the asymptotic limit, i.e. when dimensions are unboundedly large. In this case, several advanced analytic tools can be employed.

Using asymptotic characterization via the replica method, we derive a closed-form expression for the asymptotic average distortion. The validity of the results depends on the \textit{replica continuity} and \textit{replica symmetry} assumptions. Despite the lack of concrete theoretical proofs, various results in the literature confirm the validity of these assumptions in this setting.

We start representing the large-system result by defining a \textit{decoupled setting} of the under-study system. This setting is a tunable scalar system whose  average distortion, for any choice of tuning factors, is analytically calculated. Our main result indicates that for a specific choice of the tuning factors, the average distortion of the decoupled setting is equal to the  average distortion in the original system.

\begin{definition}[Decoupled setting]
\label{def:Dec}
Let $\rmp_\mH \brc{\vartheta}$ be the asymptotic singular value decomposition of the channel matrix $\mH$ whose $\rmR$-transform is $\rmR \brc{\cdot}$. For tuning parameters $c$ and $q$, define %$\tau \coloneqq {1}/{ \rmR \brc{-c} }$ and 
\begin{subequations}
\begin{align}
\tau \brc{c} &\coloneqq\frac{1}{ \rmR \brc{-c} } \label{eq:rs-4a} \\
\theta\brc{c,q} &\coloneqq \frac{1}{ \rmR \brc{-c} }, \sqrt{\frac{\partial}{\partial c} \left[ \brc{ \sigma^2 c - q } \rmR \brc{-c} \right]}. \label{eq:rs-4b}
\end{align}
\end{subequations}
Let $\xx = \psi \rms$, where $\rms$ is a uniform random variable on $\setS$, and $\psi$ is a Bernoulli random variable, independent of $\rms$, with
\begin{align}
\Pr\set{\psi=1} = 1- \Pr\set{\psi=0} = \eta.
\end{align}
Then, the decoupled output $\yy\brc{c,q}$ is defined as %for a given transmit signal $\bx$ be
\begin{align}
\yy\brc{c,q} = \xx +  \theta \brc{c,q} \zz
\end{align}
with $\zz\sim\mathcal{CN} \brc{0, 1 }$. The decoupled \ac{rls} estimation is further given by
\begin{align}
\xx^\star\brc{c,q} = \rls{\setX}{\yy\brc{c,q} \vert \tau \brc{c} },
\end{align}
where the scalar \ac{rls} recovery algorithm $\rls{\setX}{\cdot\vert\tau}$ for a given $\tau$ is defined as
\begin{align}
\rls{\setX}{\yy \vert \tau} \coloneqq \argmin_{\vv\left. \in\right. \setX_0} \left. \frac{1}{\tau}\right. \left. \abs{\yy- \vv}^2 + f_\reg \brc{\vv} \right. . \label{eq:decop_rls}
\end{align}
The decoupled distortion is moreover calculated as
\begin{align}
D_\asy \brc{c,q} \coloneqq \Ex{ F_{\rm D} \brc{\xx^\star\brc{c,q};\xx} }{ }.
\end{align}
\end{definition}

Definition~\ref{def:Dec} defines $\yy\brc{c,q}$ which is the output of a scalar \ac{awgn} channel. The input to this channel is $\xx$ whose distribution describes the empirical distribution of a significantly large transmit vector $\bx$. The noise variance of the channel is controlled with tuning factors $c$ and $q$. From this decoupled output, the estimation $\xx^\star\brc{c,q}$ is calculated, which is an \ac{rls} recovery of $\xx$ from $\yy\brc{c,q}$ via the scalar \ac{rls} estimator $\rls{\setX}{\cdot\vert\tau\brc{c,q} }$. The distortion term $D_\asy \brc{c,q}$ determines the mean distortion between decoupled estimation $\xx^\star\brc{c,q}$ and input $\xx$.

For any choice of $\setX$ and $f_\reg \brc\cdot$, the derivation of \ac{rls} estimation in the decoupled setting, i.e. $\xx^\star\brc{c,q}$, deals with solving a scalar program. Hence, in contrast to the original setting, the average distortion of \ac{rls} recovery is analytically tractable for all choices of $\setX$ and $f_\reg \brc\cdot$ in this case. Theorem~\ref{proposition:RS} indicates that for specific choices of $c$ and $q$, $D_\asy \brc{c,q}$ gives the asymptotic average distortion in the original setting. The values of $c$ and $q$, for which this equivalency happens, are given in the following theorem via a system of fixed-point equations.

\begin{theorem}
\label{proposition:RS}
Consider the multiuser \ac{mimo} setting in Section~\ref{sec:sys}, and let the receive signal be detected via an \ac{rls}-based detector with $\setX\subseteq \setC$. Assume that the technical conjectures for validity of the replica symmetric solution, including replica continuity and replica symmetry, hold. Then,
\begin{align}
\lim_{M\uparrow \infty} D \brc{M} = D_\asy \brc{c^\star,q^\star}
\end{align}
where $c^\star$ and $q^\star$ are solutions to the fixed-point equations
\begin{subequations}
\begin{align}
c \left. \theta\brc{c,q} \right.  &= \left. \tau\brc{c} \right. \Ex{ \real{\brc{\xx^\star\brc{c,q} - \xx } \zz^*} }{} \label{eq:rs-6a} \\
q &= \Ex{ \abs{\xx^\star\brc{c,q} - \xx}^2  }{}. \label{eq:rs-6b}
\end{align}
\end{subequations}
\end{theorem}
\begin{proof}
The proof follows the asymmetric decoupling property of \ac{map} estimation investigated in \cite{bereyhi2018MAP}. We start by rewriting the average distortion as 
\begin{align}
D &= \Ex{D\brc{\bi}  }{\bi}
\end{align}
where $\bi \coloneqq \left[ i_1 , \ldots, i_K\right]^\trp$ is the vector of modulation indices, and
\begin{align}
D\brc{\bi} &\coloneqq \frac{1}{M} \sum_{m=1}^M \Ex{ F_{\rm D} \brc{x^\star_m;x_m} \vert \bi }{}
\end{align}
is the average distortion for a given realization of modulation indices. Conditioned on a realization of $\bi$, the transmit signal consists of \textit{two non-identical blocks}\footnote{Note that by a \textit{block}, we mean the set of entries whose indices are in an index set $\setL\subseteq \dbc{M}$. The index set $\setL$ could be \textit{any} subset of $\dbc{M}$ and does not need to necessarily include adjacent integers.}: a block of length $\brc{1-\eta} M$ with all the entries being zero, and a block of length $\eta M$ whose entries are uniformly distributed on $\setS$. Let us denote the entry indices of the latter block, i.e. the block with $x_m\neq 0$, by $\Sp\bi \subseteq [M]$. It is hence clear that $\abs{\Sp\bi} = \eta M$. The asymmetric form of the decoupling principle, given in \cite{bereyhi2018MAP}, indicates that for any $m\in\Sp\bi$, as $M\uparrow\infty$, the conditional distribution $\rmp\brc{x^\star_m , x_m\vert \bi}$ converges to the joint distribution of $\brc{\hat{\rms} \brc{c^\star,q^\star},\rms}$, where
\begin{align}
\hat{\rms} \brc{c^\star,q^\star} \coloneqq \rls{\setX}{\rms + \theta\brc{c^\star,q^\star}\zz \vert \tau \brc{c^\star,q^\star} },
\end{align}
and $\rms$ is uniform on $\setS$. For $m\in\Sp{\bi}^{\rm C}$, $\rmp\brc{x^\star_m , x_m\vert \bi}$ converges to the distribution of $\brc{\hat{\xx}_0 \brc{c^\star,q^\star},0}$, where
\begin{align}
\hat{\xx}_0 \brc{c^\star,q^\star} \coloneqq \rls{\setX}{\theta\brc{c^\star,q^\star}\zz \vert \tau \brc{c^\star,q^\star} }.
\end{align}
Hence, $D\brc{\bi}$ is given by
\begin{subequations}
\begin{align}
D\brc{\bi} &= \frac{1}{M} \sum_{m=1}^M \Ex{ F_{\rm D} \brc{x^\star_m;x_m} \vert \bi }{}\\
&= \frac{1}{M} \left[ \sum_{m \left. \in \right. \Sp\bi} \Ex{ F_{\rm D} \brc{x^\star_m;x_m} \vert \bi }{} +  \sum_{m \left. \in \right. \Sp\bi^{\rm C}} \Ex{ F_{\rm D} \brc{x^\star_m;x_m} \vert \bi }{} \right] \\
&= \frac{1}{M} \left[ \left. \eta M \right. \Ex{ F_{\rm D} \brc{ \hat{\rms} \brc{c^\star,q^\star} ;\rms} }{}   + \left. \brc{1-\eta} M \right. \Ex{ F_{\rm D} \brc{\hat{\xx}_0 \brc{c^\star,q^\star};0} }{} \right] \\
&= \Ex{ F_{\rm D} \brc{ \xx^\star \brc{c^\star,q^\star} ;\psi\rms} }{}\\
&=D_\asy \brc{c^\star,q^\star}.
\end{align}
\end{subequations}
As $D\brc{\bi}$ is constant in $\bi$, we infer that $D=D_\asy \brc{c^\star,q^\star}$ which concludes the proof. The detailed derivations are given in Appendix~\ref{app:B}.
\end{proof}

\begin{remark}
	\label{remark:asymp}
	Comparing the result of Theorem~\ref{proposition:RS} to the asymptotic distortion of sparse recovery algorithms, e.g., \cite{bereyhi2016statistical}, one observes that these asymptotic characterizations are identical. This observation indicates that the earlier derivations based on the mismatched prior assumption of \ac{iid} sparse transmit signals closely approximate the performance.
\end{remark}

The asymptotic characterization enables us to study various aspects of \ac{masm} systems. In the sequel, we address two examples: First, we use the results to asymptotically characterize the optimal detector. We then study the recovery performance of box-constrained \ac{lasso} algorithms and discuss the optimal tuning strategy for these detectors. %The scope fo applications does not limit to the given examples and can be straightforwardly extended to multiple other applications in the literature.

\section{Application I: Bounds on Optimal Error Rate}
\label{sec:App1}
%\subsection{Optimal Detector as an RLS-based detector}
From Bayesian points of view, the optimal detector, which minimizes the probability of erroneous detection, is the \ac{map} detector which recovers ${\bx}$ as
\begin{align}
\hat\bx &= \argmax_{\bv\in\setS_0^M} \left. \rmp_\bx\brc{\bv \vert \by,\mH} \right.
\end{align}
where $\rmp_\bx\brc{\cdot \vert \by,\mH}$ denotes the posterior distribution of transmit signal $\bx$ conditioned on receive signal $\by$ and \ac{csi} $\mH$. Given $\mH$, the receiver observes an \ac{awgn} channel. Hence, it is straightforward to show that the \ac{map} detector in this case reduces to
\begin{align}
\hat\bx &= \argmin_{\bv\in\setS_0^M} \left. \norm{\by-\mH \bv}^2 - \sigma^2 \ln \rmp_\bx\brc{\bv} \right.  \label{eq:MAP}
\end{align}
where $\rmp_\bx\brc{\cdot}$ is the prior distribution of transmit signal $\bx$. The \ac{map} detector in \eqref{eq:MAP} is interpreted as an \ac{rls}-based detector in which $\setX=\setS$, the regularization function is given by \eqref{eq:prior}, and the decision function is $f_{\dec} \brc{x} = x$. 

Given the equivalent \ac{rls} form of the \ac{map} detector, it is not straightforward to formulate it, since the prior distribution $\rmp_\bx\brc\cdot$ is not of a simple form; see Example~\ref{ex:1}. Nevertheless, in the large-system limit, the detector can be approximated by a mismatched detector whose \ac{rls} formulation has a simple form. We discuss this approximated form in the sequel.

\subsection{An Approximately Equivalent Mismatched MAP Detector}
\label{Sec:MAP-Aprox}
Despite the complicated form of the exact prior distribution, $\bx$ can be approximated by an \ac{iid} distribution in the large-system limit: As $M$ grows large, the transmit signal is approximately distributed as an the \ac{iid} sequence $\tilde{\xx}_m = \tilde{\psi}_m \tilde{\rms}_m$ for $m\in[M]$, where $\tilde{\psi}_m$ is a Bernoulli random variable with\footnote{One could show that for some random codebook generations, the transmit signal in massive array case converges in distribution to this \ac{iid} sequence.}
\begin{align}
\Pr\set{\tilde{\psi}_m =1} = 1- \Pr\set{\tilde{\psi}_m =0} = \eta
\end{align}
and $\tilde{\rms}_m$ is uniformly distributed on $\setS$. The consistency of this approximation in the large-system limit can be investigated via the \ac{aep}\footnote{We skip detailed discussions in this respect, as it is out of the scope of this study.}. In the sequel, we illustrate this approximation through an example.

\begin{example}
	\label{ex:3}
Consider a scenario with $K=10$ users, each equipped with $M_{\rm u} = 16$ antennas and $L_{\rm u} = 2$ \ac{rf} chains transmitting \ac{bpsk} symbols, i.e., $\setS = \set{ \pm1 }$. For this setting, $I=6$, and thus in each symbol interval $I+L_{\rm u} S = 8$ bits of information are transmitted by each user. The codebook is generated by a random selection of $2^I = 64$ distinct antenna pairs out of the available $120$ distinct pairs. This randomly generated codebook is shared among the users and the \ac{bs}.

For this setting, we numerically realize the transmit signal $J$ times with the given codebook. Realization $j$ is denoted by $\bx\brc{j}$. Given the realizations, we determine the following two statistics for transmit entry $m$:
\begin{enumerate}
\item The empirical distribution $\tilde{\rmp}_{m} \brc{x}$, defined as
\begin{align}
\hat{\rmp}_{m} \brc{x} = \frac{1}{J} \sum_{j=1}^J \mone\set{x_m\brc{j}= x}
\end{align}
with $x_m\brc{j}$ being the $m$-th entry of $\bx\brc{j}$.
\item The empirical $\brc{\ell,t}$ joint moment function, defined as
\begin{align}
\hat{\rho}^{\ell t}_m \brc{\delta} \coloneqq \frac{1}{J} \sum_{j=1}^J x_m(j)^\ell x_{m+\delta}\brc{j}^t
\end{align}
for $\delta \in \set{-m+1, \ldots, M-m}$ and integers $\ell$ and $t$.
\end{enumerate}

The given functions are numerical evaluations of the marginal distribution of entry $m$ and its pair-wise joint moments with other entries in the transmit signal. From the classical method of moments \cite{akhiezer1965classical}, one could argue that if the analytical terms for the marginal distributions, and all joint moments\footnote{This means not only the pair-wise joint moments, but also joint moments of more entries.} are equal to their corresponding functions in sequence $\set{\tilde{\xx}_m}$, for $m\in M$ and all moment exponents; then, the transmit entries and sequence $\set{\tilde{\xx}_m}$ have identical distributions. This constraint is however intractable to be checked, even numerically. We hence consider only the given empirical measures and compare them to their corresponding metrics in sequence $\set{\tilde{\xx}_m}$; namely, to distribution of $\tilde{\xx}_m$, shown by $\rmp_{\tilde{\xx}_m}\brc{x}$, and joint moment $\tilde{\rho}^{\ell t}_m \brc{\delta}$, defined as
\begin{align}
\tilde{\rho}^{\ell t}_m \brc{\delta} \coloneqq \Ex{ \tilde{\xx}_m^\ell \tilde{\xx}_{m+\delta}^t }{}.
\end{align}
	
Fig.~\ref{fig:prior} and \ref{fig:cor} show numerical results for $J=10^5$ realizations considering transmit entry $m=80$. In Fig.~\ref{fig:prior}, the empirical distribution is sketched showing close consistency with $\rmp_{\tilde{\xx}_m}\brc{x}$, even for moderate dimensions. 

Fig.~\ref{fig:cor} further shows joint moment functions for multiple choices of $\ell$ and $t$. In theory, $\tilde{\rho}^{\ell t}_m$ has three different forms, depending on the values of $\ell$ and $t$. Fig.~\ref{fig:cor} contains an empirical sample of each form. It is observed that the joint moments closely track $\tilde{\rho}^{\ell t}_m \brc{\delta}$ for the corresponding values of $\ell$ and $t$.
\end{example}

\begin{figure}[!t]
\centering
\begin{tikzpicture}

\begin{axis}[%
width=4in,
height=2.8in,
at={(1.989in,1.234in)},
scale only axis,
xmin=-1.5,
xmax=1.5,
xtick={-1,  0,  1},
xlabel style={font=\color{white!15!black}},
xlabel={$x$},
ymin=0,
ymax=1,
ytick={   0, 0.25,  0.5, 0.75,    1},
ylabel style={font=\color{white!15!black}},
ylabel={$\hat{\rmp}_{m} \brc{x}$},
axis background/.style={fill=none},
legend style={legend cell align=left, align=left, draw=white!15!black}
]
\addplot[forget plot,ycomb, color=blue, draw=none, mark size=5.0pt, mark=x, mark options={solid, blue}] table[row sep=crcr] {%
-1	0.06284\\
0	0.87604\\
1	0.06112\\
};
\addplot[forget plot, color=white!15!black, draw=none] table[row sep=crcr] {%
-1.5	0\\
1.5	0\\
};
%\addlegendentry{ex}

\addplot[forget plot,ycomb, color=red, line width=1.0pt, mark size=3.5pt, mark=square, mark options={solid, red}] table[row sep=crcr] {%
-1	0.0625\\
0	0.875\\
1	0.0625\\
};
\addplot[forget plot, color=white!15!black, line width=1.0pt] table[row sep=crcr] {%
-1.5	0\\
1.5	0\\
};
%\addlegendentry{th}

\end{axis}
\end{tikzpicture}%
\caption{Empirical marginal distribution of transmit entry $m=80$. Numerical simulations are denoted by blue crosses closely tracking $\rmp_{\tilde{\xx}_m}\brc{x}$ shown by red squares. Here, the number of users is $K=10$. Each user has $M_{\rm u} = 16$ antennas and transmits \ac{bpsk} symbols over $L_{\rm u} = 2$ active antennas.}
\label{fig:prior}
\end{figure}
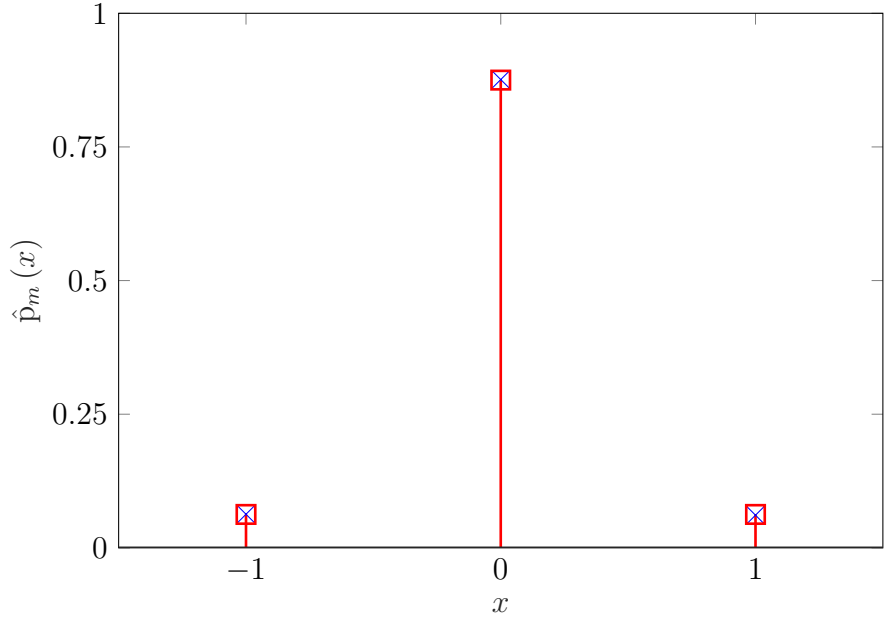

\begin{figure*}[!t]
\centering
\input{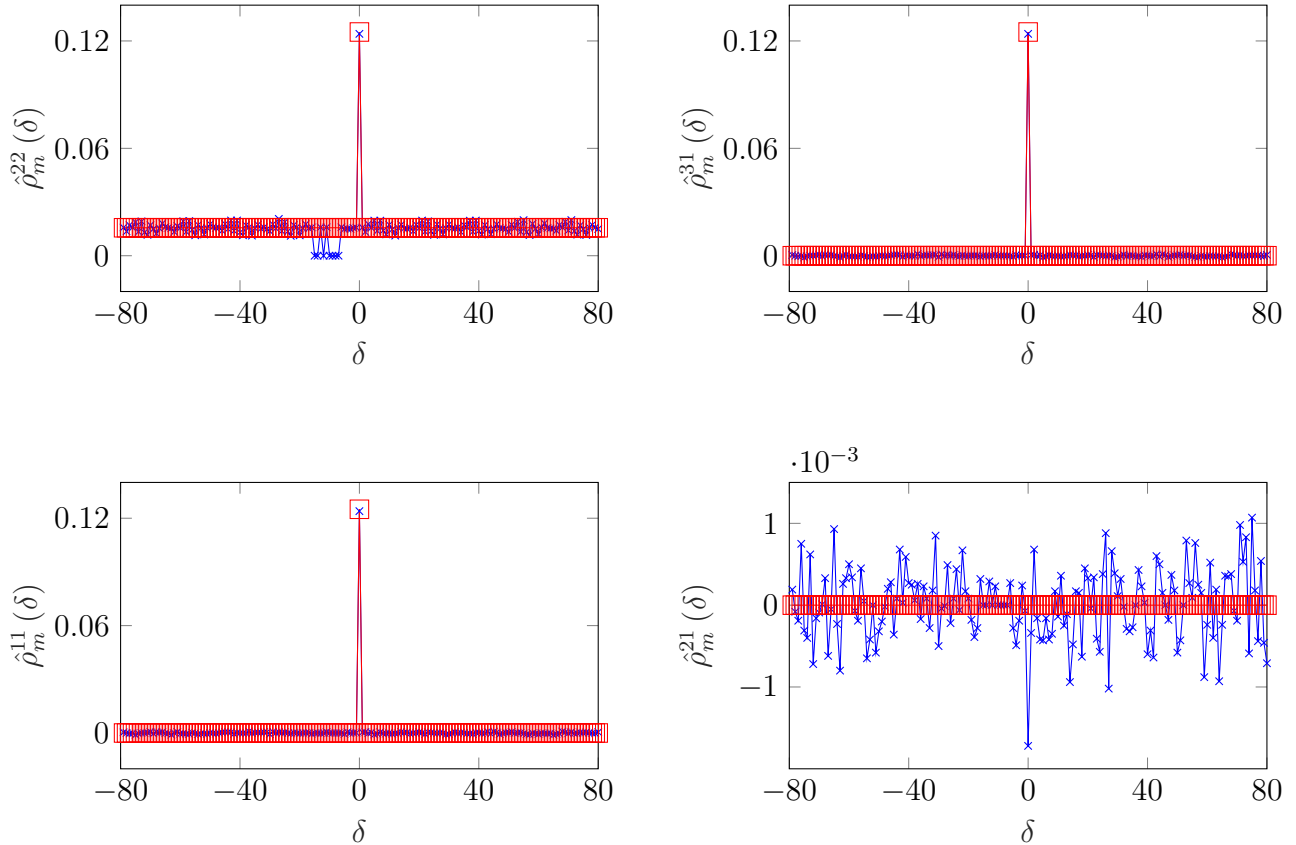}
\caption{Joint moment function of transmit entry $m=80$ for multiple choices of $\ell$ and $t$. The number of users is $K=10$, each having $M_{\rm u} = 16$ antennas and $L_{\rm u} = 2$ \ac{rf} chains. The \ac{bpsk} constellation is considered for transmission. The red squares denote $\tilde{\rho}^{\ell t}_m \brc{\delta}$ while blue crosses show numerical results.}
\label{fig:cor}
\end{figure*}

The observations in Example~\ref{ex:3} numerically justify the fact that the \ac{iid} sequence $\set{\tilde{\xx}_m}$ is a good approximation of the true transmit signal. Considering this finding, we approximate the optimal performance by the performance of a \textit{mismatched} \ac{map} detector which postulates the prior distribution
\begin{align}
\rmp_{\bx, \mathrm{pos}} \brc{\bv} = \prod_{m=1}^M \rmp_{\xx} \brc{v_m}, \label{eq:postulate}
\end{align}
for the transmit signal. Here, $\rmp_{\xx} \brc{\cdot}$ denotes the distribution of $\tilde{\xx}_m$\footnote{This distribution is equivalent to the distribution of the decoupled input $\xx$ given in Definition~\ref{def:Dec}.} and is given by
\begin{align}
\rmp_{\xx} \brc{v} = \brc{1-\eta} \mone\set{ v = 0 } + \left. 2^{-S} \eta \right. \mone\set{ v \neq 0 }.
\end{align}
Invoking Theorem~\ref{proposition:RS}, we characterize the performance of this mismatched detector in the sequel.

\subsection{Asymptotic Performance of the Mismatched MAP Detector}
\label{sec:MAP}
For the postulated prior distribution in \eqref{eq:postulate}, we have
\begin{subequations}
\begin{align}
-\sigma^2 \loge \rmp_{\bx, \mathrm{pos}} \brc{\bv} &= - \sigma^2 \sum_{m=1}^M \loge \rmp_{\xx} \brc{v_m}\\
&= a \norm{\bv}_0 + b
\end{align}
\end{subequations}
assuming that $0 < \eta < 1$, where $a$ and $b$ are given by
\begin{subequations}
\begin{align}
a &\coloneqq \sigma^2 \left[ S \loge  2 + \loge \brc{1-\eta} - \loge {\eta} \right] \\
b &\coloneqq - \sigma^2 \loge \brc{1-\eta}.
\end{align}
\end{subequations}
Hence, the regularization function of the mismatched detector is\footnote{Note that constant $b$ does not play any role in the optimization problem. Thus, it is dropped.}
\begin{align}
f_\reg \brc{\bv} = a \norm{\bv}_0.
\end{align}
$f_\reg \brc{\bv}$ regularizes the least-squares, i.e. $ \norm{\mH \bv-\by}^2 $, with respect to the sparsity of the transmit signal imposed via index modulation. With this regularization function, the mismatched detector reduces to the so-called \textit{$\ell_0$-norm minimization} algorithm with regularization parameter $a$, considered in compressive sensing for sparse recovery.

Using Theorem~\ref{proposition:RS}, we now determine the average error rate of the mismatched detector in the large-system limit. To this end, we first derive an analytical expression for the error rate of the decoupled setting. The asymptotic error rate is then calculated from the derived expression using the solution of the fixed-point equations. Detailed derivations are given below:
\subsubsection*{Decoupled setting}
The decoupled \ac{rls} estimation for the mismatched detector is determined by setting 
\begin{align}
f_\reg \brc{\vv} = \left. a \right. \mone\set{ \vv \neq 0 }
\end{align}
in \eqref{eq:decop_rls}. Moreover, as $f_{\dec}\brc{x}=x$, we can conclude that 
\begin{align}
\hat{\xx} \brc{c,q} = f_{\dec}\brc{ \xx^\star \brc{c,q} } = \xx^\star \brc{c,q}.
\end{align}
Hence, we have
\begin{align}
\hat{\xx} \brc{c,q} =
\begin{cases}
0 & U_{\max} \brc{c,q} \leq \tau\brc{c} a\\
s_{\max} \brc{c,q} & U_{\max} \brc{c,q} > \tau\brc{c} a
\end{cases}\label{eq:decoup_setting_I}
\end{align}
where $U_{\max} \brc{c,q}$ and $s_{\max} \brc{c,q}$ are given by
\begin{subequations}
\begin{align}
U_{\max} \brc{c,q} &\coloneqq \max_{s\left. \in  \right. \setS} \left[ 2 \real{ \yy^* \brc{c,q} s } - \abs{s}^2 \right] \\
s_{\max} \brc{c,q} &\coloneqq \argmax_{s\left. \in  \right. \setS} \left[ 2 \real{ \yy^* \brc{c,q} s } - \abs{s}^2 \right] .
\end{align}
\end{subequations}
\subsubsection*{Asymptotic average error}
We set the distortion function to \eqref{eq:error_func}. The asymptotic error rate is then given by the decoupled distortion when $c$ and $q$ are set to $c^\star$ and $q^\star$, respectively. For the decoupled estimation in \eqref{eq:decoup_setting_I}, $D_{\asy} \brc{c,q}$ is given by
\begin{align}
D_{\asy} \brc{c,q} &= \Ex{ \mone\set{\hat{\xx} \brc{c,q}  \neq \xx } }{} = 1- P_{\rm C} \brc{c,q}
\end{align}
where $P_{\rm C} \brc{c,q}$ is given by
\begin{subequations}
\begin{align}
P_{\rm C} \brc{c,q} \coloneqq& \ \Ex{ \mone\set{\hat{\xx} \brc{c,q}  = \xx } }{}\\
=& \  \brc{1-\eta} G_0 \brc{c,q} + 2^{-S} \eta \sum_{s \left. \in \right. \setS } G_s \brc{c,q}
\end{align}
\end{subequations}
with $G_s \brc{c,q}$ being defined as
\begin{align}
G_s \brc{c,q} &\coloneqq \Pr\set{ {\hat{\xx} \brc{c,q}  = s } \vert \yy\brc{c,q} = s+\theta \brc{c,q} \zz }
\end{align}
for $s\in\setS_0$. Noting that $\zz\sim\mac\man\brc{0,1}$, $G_s \brc{c,q}$ is a Gaussian integral, which can be straightforwardly calculated.
\subsubsection*{Solving the fixed-point equations}
%To find $c^\star$ and $q^\star$, we consider the definition of $\yy\brc{c,q}$. Noting that $\xx = \sqrt{P} \psi $, it reads
%\begin{align}
%\yy\brc{c,q} = 
%\begin{cases}
%\theta\brc{c,q} \zz & \psi = 0\\
%\sqrt{P}+\theta\brc{c,q} \zz & \psi = 1
%\end{cases}.
%\end{align}
%The decoupled soft estimation is hence a function of random variables $\psi$ and $\zz\sim\mac\man\brc{0,1}$. As a result, the fixed-point equations are written as
By substituting the decoupled estimation into Theorem~\ref{proposition:RS}, the fixed-point equations are given by
\begin{subequations}
\begin{align}
c \left. \theta\brc{c,q} \right.  &= \left. \tau\brc{c} \right. \mac\brc{c,q}  \label{Fix1} \\
q &= \mae\brc{c,q}. \label{Fix2}
\end{align}
\end{subequations}
where $\mac\brc{c,q}$ and $\mae\brc{c,q}$ are given by
\begin{subequations}
\begin{align}
\mac\brc{c,q}  \coloneqq& \hspace*{1mm} \Ex{ \real{\brc{\hat{\xx}\brc{c,q} - \xx } \zz^*} }{} \\
\mae\brc{c,q}  \coloneqq& \hspace*{1mm}  \Ex{ \abs{\xx^\star\brc{c,q} - \xx}^2  }{}.
\label{eq:Cor}
\end{align}
\end{subequations}
Similar to the decoupled distortion, $\mac\brc{c,q}$ and $\mae\brc{c,q}$ are written as sums of Gaussian integrals. For instance,
\begin{align}
\mae\brc{c,q} = \brc{1-\eta} \mae_0 \brc{c,q} + 2^{-S} \eta \sum_{s \left. \in \right. \setS } \mae_s \brc{c,q}
\end{align}
where $\mae_s \brc{c,q}$ is defined as
\begin{align}
\hspace*{-1.5mm}\mae_s \brc{c,q} \coloneqq  \Ex{ \left. \abs{\hat{\xx}\brc{c,q} - \xx}^2 \right\vert \yy\brc{c,q} \hspace*{-.5mm} = \hspace*{-.5mm} s\hspace*{-.5mm} + \hspace*{-.5mm} \theta \brc{c,q} \zz  }{\zz}.
\end{align}
%and $\mae\brc{c,q}$ reads
%\begin{subequations}
%\begin{align}
%\mae\brc{c,q}  \coloneqq& \hspace*{1mm}  \Ex{ \abs{\xx^\star\brc{c,q} - \xx}^2  }{}\\
%=& \hspace*{1mm} \brc{1-\eta}\Ex{ \abs{\xx^\star\brc{c,q} - \xx}^2 \vert \psi = 0 }{\zz} \nonumber \\
%&+ \left. \eta \right. \Ex{ \abs{\xx^\star\brc{c,q} - \xx}^2 \vert \psi = 1 }{\zz}. \label{eq:Energy}
%\end{align}
%\end{subequations}
Consequently, to solve the fixed-point equations, we calculate $\mac\brc{c,q}$ and $\mae\brc{c,q}$ explicitly from the Gaussian integrals, and plug them into \eqref{Fix1} and \eqref{Fix2}. The resulting equations are then solved numerically\footnote{An alternative approach is to iteratively find the stability point of the corresponding \textit{replica simulator}; see \cite{bereyhi2016statistical} for detailed discussions.}.

\subsection{Numerical Investigations}
The achievable average error rate for the mismatched \ac{map} detector is plotted against the $\snr$ in Fig.~\ref{fig:MAP_bound}. Here, the number of transmit antennas and \ac{rf} chains at user terminals are set to $M_{\rm u} = 8$ and $L_{\rm u} = 1$, respectively. Hence, the activity ratio evaluates to $\eta = L_{\rm u} / M_{\rm u} = 1/8$. The system load is $\alpha = 1/4$ meaning that there are four receive antennas per user terminal at the \ac{bs}, i.e. $N/K = 4$. The figure shows the error rate for three different scenarios with different constellation sets. Namely,
\begin{itemize}
\item 4-\ac{qam} transmission, in which
\begin{align}
\setS_{\rm QAM} = \set{ \pm\sqrt{\frac{P}{2}} \pm \rmj \sqrt{\frac{P}{2}}  },
\end{align}
\item \ac{bpsk} transmission with
\begin{align}
\setS_{\rm BPSK} = \set{ \pm\sqrt{P} },
\end{align}
\item \ac{ssk} transmission for which
\begin{align}
\setS_{\rm SSK} = \set{ \sqrt{P} }.
\end{align}
\end{itemize}

For the shown results the \ac{iid} Rayleigh fading model is assumed in which the channel gains are \ac{iid} zero-mean complex Gaussian random variables with variance $1/M$. The $\rmR$-transform of the asymptotic squared singular values of this channel matrix is given by \cite{tulino2004random} 
\begin{align}
\rmR\brc{c} = \frac{\xi^{-1}}{1-c} \; . \label{eq:R-transform}
\end{align}
To sweep over the $x$-axis, we set $P=1$, and $\sigma^2$ is appropriately scaled for a given $\snr$. %assuming that unit symbol interval, i.e., $T=1$. $E_{\rm b} / N_0$ is hence calculated as $P/N_{\rm Bits} \sigma^2$, where $N_{\rm Bits}$ denotes the number of bits transmitted in each symbol interval via the given scheme.

\begin{figure}[!t]
\centering
\input{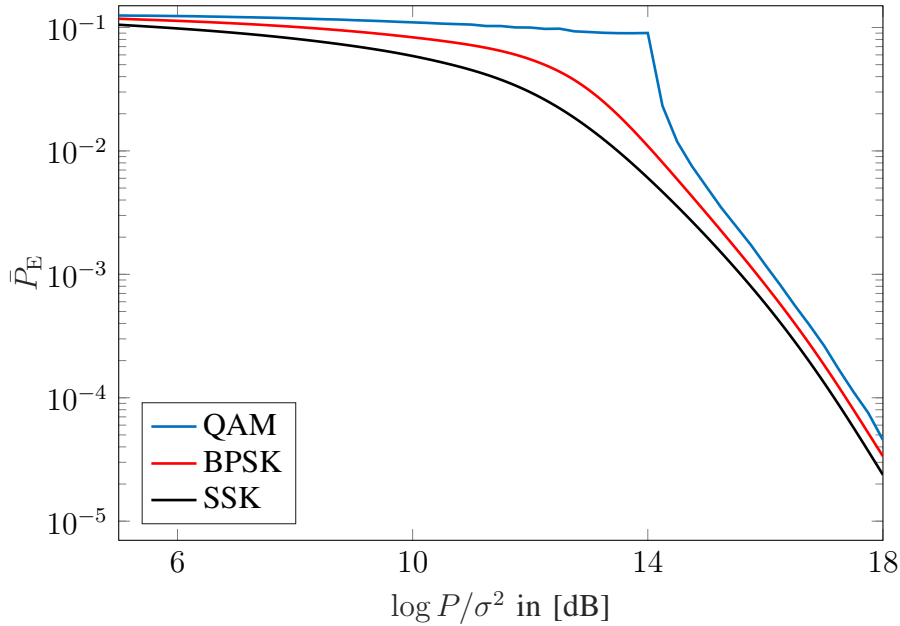}
\caption{Asymptotic error rate for the mismatched MAP estimator vs. $\snr$. Here, $M_{\rm u} = 8$ and $L_{\rm u} = 1$. The system load is set to $\alpha = 1/4$. The curves are good approximations of the minimum error rate.}
\label{fig:MAP_bound}
\end{figure}

The figure shows that the performance of the detector degrades, as the constellation size increases. This is an intuitive observation, since the energy efficiency reduces if the constellation size is raised at a fixed power. Following earlier discussions, Fig. \ref{fig:MAP_bound} is considered as a good approximation of the \textit{minimum} error rate bound.

\section{Application II: Tuning LASSO-Type Detectors}
\label{sec:App2}
%The large-system characterization enables us investigating \ac{masm} in several respects. Here, we name few applications.
%\begin{enumerate}
%\item With $\setX=\setS$ and setting the regularization term to be an indicator function, the \ac{rls}-based detector reduces to the optimal Bayesian detector, i.e. \ac{ml} detector. The average error rate of this detector is hence derived in the asymptotic regime, using the result in the previous section. Although the \ac{ml} detector is not implementable for large dimensions in practice, the determined error rate can be considered as a tight lower bound on the achievable error rate in an \ac{masm} setting.
%\item 
In practice, suboptimal detectors are used for data recovery in \ac{masm} \ac{mimo} systems which are often \ac{rls}-based detectors. The performance of these detectors usually depends on a set of parameters which need to be effectively tuned. The tuning task aims to find these parameters, such that a desired distortion metric is minimized. 

In general, the optimal parameters of a detector depend on the factors at which the system operates, e.g. \ac{snr}. Therefore, the detector should be re-tuned frequently over time. A direct approach for tuning can hence burden the system computationally. Motivated by this, we propose a tuning strategy based on the large-system characterization.

\begin{tuning}
Consider a distortion metric, with respect to which an \ac{rls}-based detector is to be tuned. Using Theorem~\ref{proposition:RS}, derive the asymptotic limit of this metric for generic tuning parameters. Then determine the parameters, such that the given metric is minimized.
\end{tuning}

The advantage of this approach is that the optimal values of tunable parameters are analytically determined in terms of system factors. Hence, it imposes almost no additional computational complexity to the system. On the other hand, noting that the given detector is tuned asymptotically, the efficiency of the approach is questionable. In this section, we demonstrate how efficient this approach performs by considering \textit{\ac{lasso}-type} detectors for signal recovery in \ac{masm} systems. 

\subsection{LASSO-Type Detectors}
A \ac{lasso}-type detector corresponds to an \ac{rls}-based detection algorithm in which the regularization function is linearly proportional to the $\ell_1$-norm. This means that in these detectors
\begin{align}
f_{\reg} \brc{\bv} = \lambda \norm{\bv}_1
\end{align}
for some $\lambda \in \setR^+$ which is referred to as the regularization parameter.

\ac{lasso} is known to be an effective regularization for sparse recovery \cite{candes2006robust}. From computational points of view, $\ell_1$-norm regularization results in a convex objective function. Thus, for convex choices of $\setX$, the \ac{rls} optimization is tractably solved via convex programming. 

\ac{lasso}-type detectors for \ac{sm} relax $\setS_0$ into a convex set $\setX$ in order to get a computationally tractable recovery algorithm. Depending on the choice of $\setX$, there are various \ac{lasso}-type detection algorithms. These algorithms are roughly divided into two types: 
\begin{enumerate}
\item \textit{Classic \ac{lasso}} in which $\setS_0$ is relaxed into either $\setC$ or $\setR$, depending on the constellation.
\item \textit{Box-\ac{lasso}} in which $\setX$ is  a convex subset of the complex plane comprising $\setS_0$.
\end{enumerate}

In contrast to the classic \ac{lasso}, box-\ac{lasso} detectors are not investigated widely in the literature. This follows the observation that the box relaxations give negligible gains in several applications of sparse recovery. The study in \cite{atitallah2017box} has however shown that for spatially modulated signals, the box-\ac{lasso} algorithm achieves a considerable enhancement in some scenarios.
%An example of these detectors is the \textit{box-\ac{lasso}} detector, considered in \cite{atitallah2017box}. Our investigations show that the large-system characterization is tightly consistent with the performance of \ac{lasso}-type detectors even for settings with not-so-large dimensions. The results of this section extends the study in \cite{atitallah2017box} to the wider class of unitarily invariant channel matrices.
\subsection{Tuning Task for LASSO-Type Detectors}
The performance of \ac{lasso}-type detectors depends on the regularization parameter $\lambda$. Although the detection algorithm performs effectively for a proper choice of $\lambda$, setting $\lambda$ to some other value can significantly degrade the performance. This is observed through a numerical experiment:

\begin{example}
\label{ex:4}
$K=10$ users transmit unit-power \ac{ssk} signals, i.e., $S = 0$ and $P=1$. Each user terminal is equipped with $M_{\rm u} = 8$ antennas and $L_{\rm u} = 1$ \ac{rf} chain. The uplink channels experience \ac{iid} Rayleigh fading with zero-mean and variance $1/M$, and $\log \snr = 11$ dB. The \ac{bs} has $N=160$ antennas and employs a classic \ac{lasso} detector which determines 
\begin{align}
\bx^\star = \argmin_{\bv\in\setX_0^M} \left. \norm{\by-\mH \bv}^2 + \lambda \norm{\bv}_1 \right.
\end{align}
for $\setX = \setR$. The transmit signal is then recovered by setting all entries of $\bx^\star$ whose values are less than $\epsilon = 0.5$ to zero, and the rest to one. 

In practice, this detector is tuned such that the average \ac{mse} of the soft estimation is minimized, i.e, the average distortion with the distortion function given in \eqref{eq:MSE_Dist}. Choosing the \ac{mse} as the tuning metric follows its high robustness. For this setting, we plot in Fig.~\ref{fig:LASSO_example} the average \ac{mse} against $\lambda$. The simulations are averaged over $10^3$ independent realizations of the setup. As the figure shows, the optimal choice for the regularization parameter is approximately $\lambda^\star \approx 0.56$ at which $\log \mse = - 20.73$ dB. However, in the case of using a mismatched parameter $\lambda=0.06$, the \ac{mse} increases by $4$ dB. 
\end{example}

\begin{figure}[!t]
\centering
\definecolor{mycolor1}{rgb}{0.00000,0.44700,0.74100}%
\begin{tikzpicture}

\begin{axis}[%
scaled ticks=false,
width=4in,
height=2.8in,
at={(1.984in,1.234in)},
scale only axis,
xmin=0,
xmax=.6,
xtick={0.06,0.2,.4, .56},
xticklabels={{$0.06$},{$0.2$},{$0.4$},{$0.56$}},
xlabel style={font=\color{white!15!black}},
xlabel={$\lambda$},
ymin=-21.2,
ymax=-15.5,
ytick={-20.73,-16.73},
yticklabels={ {$-20.73$} , {$-16.73$}},
ylabel style={font=\color{white!15!black}},
ylabel={$\log \mse$ in [dB]},
axis background/.style={fill=white}
]
\addplot [color=mycolor1, line width=1.0pt, mark size=2.2pt, mark=*, mark options={solid, mycolor1}, forget plot]
  table[row sep=crcr]{%
0.01	-15.9169682443427\\
0.06	-16.730758616177\\
0.11	-17.4551240970641\\
0.16	-18.1387520706127\\
0.21	-18.7382484171431\\
0.26	-19.3135198098367\\
0.31	-19.7261307556428\\
0.36	-20.1772728284571\\
0.41	-20.3808743549023\\
0.46	-20.5407731109134\\
0.51	-20.7174755733688\\
0.56	-20.7344550704232\\
};
\addplot [color=black, dashed, forget plot]
  table[row sep=crcr]{%
0	-20.7344550704232\\
0.56	-20.7344550704232\\
};
\addplot [color=black, dashed, forget plot]
  table[row sep=crcr]{%
0	-16.730758616177\\
0.06	-16.730758616177\\
};
\addplot [color=black, dashed, forget plot]
  table[row sep=crcr]{%
0.06	-16.730758616177\\
0.06	-25\\
};
\addplot [color=black, dashed, forget plot]
  table[row sep=crcr]{%
0.56	-20.7344550704232\\
0.56	-25\\
};
\end{axis}
\end{tikzpicture}%
\caption{MSE of a classic LASSO detector vs. $\lambda$. In the underlying scenario, $K=10$ users, each equipped with $M_{\rm u} = 8$ antennas and a single \ac{rf} chain, transmit SSK signals. The \ac{bs} has $N=160$ antennas and receives the signals at $\log \snr = 11$ dB.}
\label{fig:LASSO_example}
\end{figure}
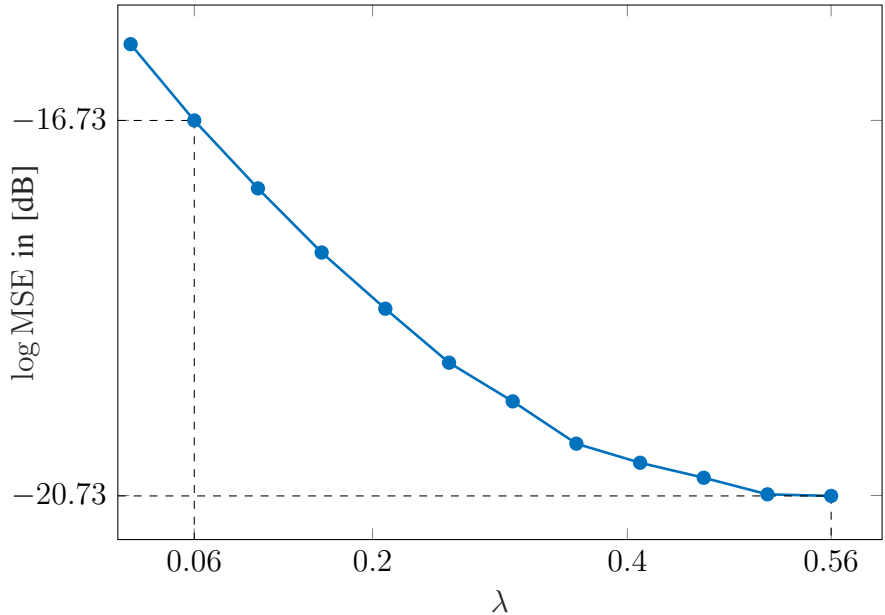

The observation in Example~\ref{ex:4} indicates the importance of accurate tuning for \ac{lasso}-type detectors. For sake of simplicity, we investigate our proposed approach in the sequel, by considering the special case of \ac{ssk} transmission. %The derivations can be extended to other cases as well.

\subsection{Analysis of LASSO-Type Detectors for SSK Transmission}
In this section, we investigate the \ac{lasso}-type detection of \ac{ssk} transmission, i.e., $\setS = \{\sqrt{P} \}$, by considering the following box-constrained \ac{lasso} scheme:
\begin{itemize}
\item The regularization function is
\begin{align}
f_{\reg} \brc{\bv} = \lambda \norm{\bv}_1
\end{align}
for some \textit{regularization parameter} $\lambda$ which is tunable. 
\item Set $\setX$ is set to
\begin{align}
\setX = \left[ -\ell , u \right]
\end{align}
for some $\ell \geq 0$ and $u \geq \sqrt{P}$.
\item The decision function is given by
\begin{align}
f_{\dec} \brc{x} = \left. \sqrt{P} \right. \mone\set{x\geq \epsilon} =
\begin{cases}
\sqrt{P} &x \geq \epsilon\\
0		 &x < \epsilon
\end{cases}\label{eq:dec_box}
\end{align}
for some given \textit{threshold} $\epsilon$.
\end{itemize}

The given detector can be observed as a mismatched \ac{map} detector which postulates the signal prior distribution to be
\begin{align}
\rmp_{\bx, \mathrm{pos}} \brc{\bv} = \frac{1}{Z}
\begin{cases}
\exp\set{-\norm{\bv}_1} &\bv \in \setH \\
0 &\text{otherwise}
\end{cases}
\end{align}
with constant $Z = 2-\exp\set{-\ell} - \exp\set{-u}$, and $\setH$ denoting an $M$-dimensional hypercube constructed by limiting each di-mension to $\setX_0 = \left[ -\ell , u \right]$, i.e.,
\begin{align}
\setH = \set{ \bv \in \setR^N: -\ell \leq v_m \leq u \ \text{ for } m \in \dbc{M} }.
\end{align}

\begin{figure}[!t]
	\centering
	\input{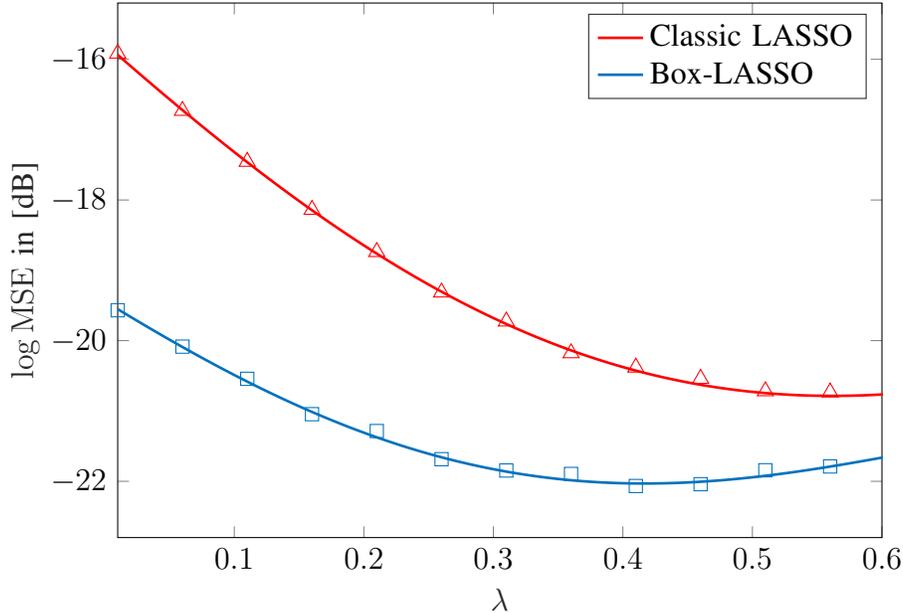}
	\caption{Asymptotic \ac{mse} against regularization parameter for Example~\ref{ex:4} when a box-\ac{lasso} algorithm with support $\setX=\dbc{0,1}$ is employed for detection. The markers show simulation results for the sizes given in Example~\ref{ex:4}.}
	\label{fig:1}
\end{figure}

For this particular setting, the asymptotic performance was characterized in \cite{atitallah2017box} for \ac{iid} Gaussian channel matrices\footnote{Entries of the channel matrix are however assumed to be real in \cite{atitallah2017box}.} via the convex Gaussian min-max theorem. Following numerical simulations and universality results \cite{korada2011applications,oymak2017universality}, it was then conjectured that the asymptotic results extend further beyond \ac{iid} Gaussian matrices. This conjecture was partially investigated in \cite{alrashdi2019precise} by providing some analysis.

In the sequel, we derive the asymptotic average distortion invoking Theorem~\ref{proposition:RS}. The results are valid for right unitarily invariant channel matrices including the formerly studied \ac{iid} Gaussian matrix, as well as other matrices whose corresponding performances were conjectured in \cite{atitallah2017box}.

\subsubsection*{Decoupled setting}
For this setting, the decoupled input $\xx$ is given by $\xx = \sqrt{P} \psi $ with $\psi$ being a Bernoulli random variable described by Definition~\ref{def:Dec}. The decoupled soft estimation is given in terms of $\yy\brc{c,q}$ as
	\begin{align}
	\xx^\star\brc{c,q} = 
	\begin{cases}
	u &\qquad \phantom{-\dfrac{\tau\brc{c}\lambda}{2} - \ell  \leq} \ \yy\brc{c,q} \geq \dfrac{\tau\brc{c}\lambda}{2}   + u \vspace*{2mm}\\
	\yy\brc{c,q} - \dfrac{\tau\brc{c}\lambda}{2} &\qquad \phantom{- - \ell}\dfrac{\tau\brc{c}\lambda}{2}  \leq \yy\brc{c,q} \leq \dfrac{\tau\brc{c}\lambda}{2} + u\vspace*{2mm}\\
	0 &\qquad \phantom{-\ell} -\dfrac{\tau\brc{c}\lambda}{2}  \leq \yy\brc{c,q} \leq \dfrac{\tau\brc{c}\lambda}{2} \vspace*{2mm}\\
	\yy\brc{c,q} + \dfrac{\tau\brc{c}\lambda}{2} &\qquad -\dfrac{\tau\brc{c}\lambda}{2}-\ell  \leq \yy\brc{c,q} \leq -\dfrac{\tau\brc{c}\lambda}{2}\vspace*{2mm}\\
	-\ell &\qquad \phantom{-\dfrac{\tau\brc{c}\lambda}{2}-\ell  \leq} \ \yy\brc{c,q} \leq -\dfrac{\tau\brc{c}\lambda}{2} - \ell
	\end{cases}
	\label{eq:dec_box_lasso}
	\end{align} 

\subsubsection*{Fixed-point Equations}
Following the approach in Section \ref{sec:MAP}, the fixed-point equations are derived as sums of Gaussian integrals which are straightforward to calculate. 

\subsection{Numerical Investigations}
\begin{figure}[!t]
	\centering
	\input{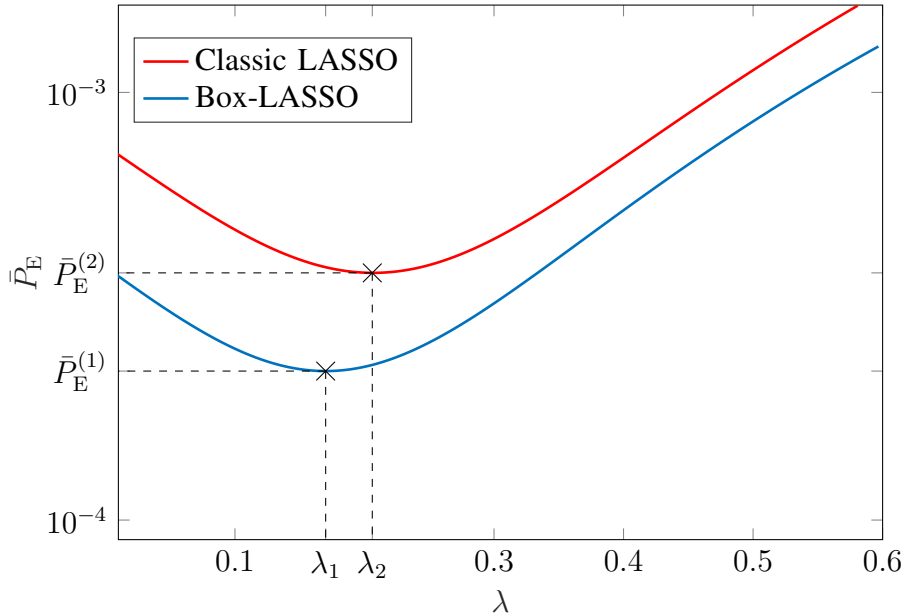}
	\caption{Asymptotic error rate vs. $\lambda$ for Example~\ref{ex:4}.}
	\label{fig:2}
\end{figure}
We consider the setting of Example~\ref{ex:4} but replace the classic \ac{lasso} detector is replaced by a box-constrained \ac{lasso} detector with support $\setX = \dbc{0,1}$. For this example, we plot the asymptotic average \ac{mse} against $\lambda$ in Fig.~\ref{fig:1}, using Theorem~\ref{proposition:RS}, where we set $\alpha = K/N = 1/16$ and $\xi = M_{\rm u} \alpha = 0.5$. In addition, the result for classic \ac{lasso} detector is shown for comparison. As the figure shows, the asymptotic results closely track finite-dimensional numerical simulations which are averaged over $10^3$ realizations.

We now tune these \ac{lasso}-type detectors using our tuning approach. To this end, we first plot the asymptotic error rate for the detectors against regularizer $\lambda$ in Fig.~\ref{fig:2}. As the figure shows, the error rate is minimized for the box-\ac{lasso} and classic \ac{lasso} at
\begin{subequations}
	\begin{align}
	\lambda_1 &\approx 0.17,\\
	\lambda_2 &\approx 0.206,
	\end{align}
\end{subequations}
with values 
\begin{subequations}
	\begin{align}
	\bar{P}_{\rm E}^{(1)} &\approx 1.9 \times 10^{-4},\\
	\bar{P}_{\rm E}^{(2)} &\approx 3.1 \times 10^{-4},
	\end{align}
\end{subequations}
respectively. We now repeat this procedure while sweeping $\log \snr$ from $5$ dB to $13$ dB. The optimized error rates for both the detectors are sketched in Fig.~\ref{fig:3} against $\snr$. For sake of comparison, the optimal error bound derived in Section~\ref{Sec:MAP-Aprox} is further sketched. It is observed that the box-\ac{lasso} detector achieves roughly an error rate in between of the optimal rate and the classic \ac{lasso} rate. For sake of comparison, the tuned detectors are simulated numerically for the sizes given in Example~\ref{ex:4}. The results are shown by the markers in the figure. The simulations show close consistency with the asymptotic results.

\begin{figure}[t]
	\centering
	\input{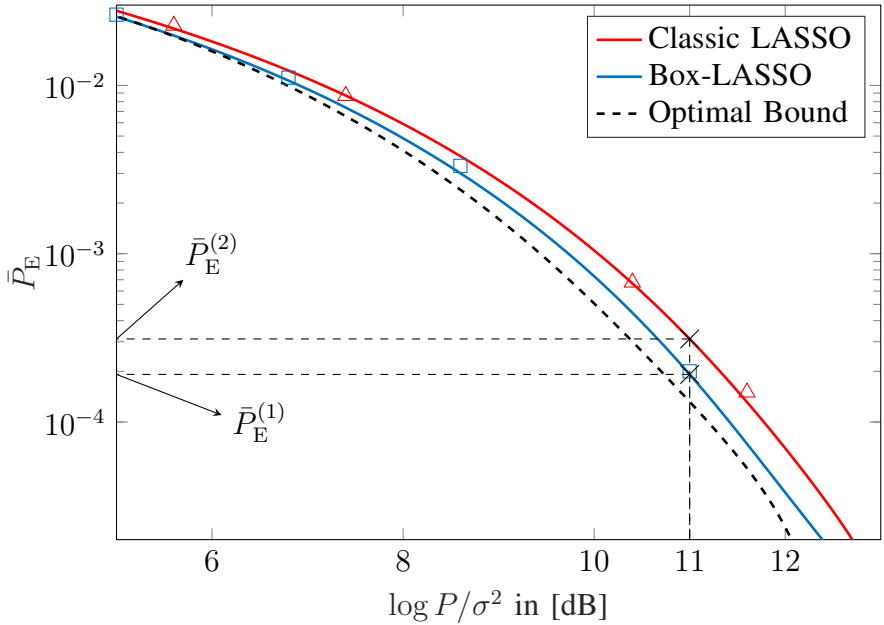}
	\caption{The minimum achievable error rate for box-\ac{lasso} detection vs. $\snr$.}
	\label{fig:3}
\end{figure}

We finally show the tuned regularizer in terms of $\snr$ in Fig.~\ref{fig:4}. Such a curve can be seen as a dictionary, which is derived analytically prior to system setup. During the transmission, the regularizer is continuously updated to the tuned value corresponding to the operating point at which the system operates.
\begin{figure}[t]
	\centering
	\input{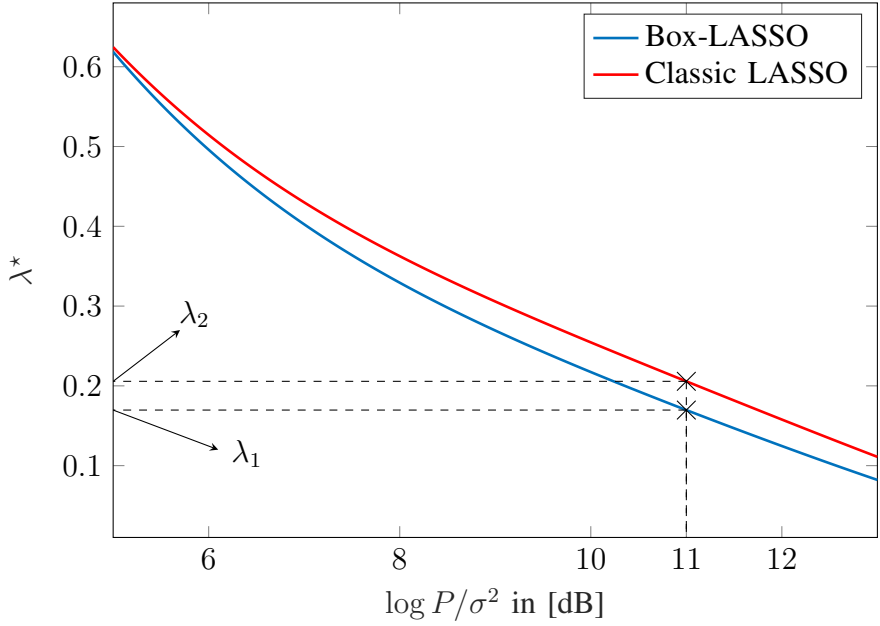}
	\caption{Optimal regularization parameter for box-\ac{lasso} detection vs. $\snr$.}
	\label{fig:4}
\end{figure}

\section{Conclusions}
\label{sec:Conc}
%This work has studied the large-system performance of \ac{rls}-based algorithms when used for the detection of spatially modulated signals. 
Considering the precise \textit{non-\ac{iid}} model of the transmit signals, the analytical results have demonstrated that the system asymptotically behaves identically to a scenario in which \ac{iid} sparse signals with similar sparsity are transmitted. This finding validates earlier analyses in the literature, e.g., \cite{atitallah2017box}, and was further observed through numerical investigations of some special cases in Sections~\ref{sec:App1} and \ref{sec:App2}.

Together with our earlier investigations in \cite{gade2019fair}, this work draw a clear picture of \ac{sm} \ac{mimo} systems and their use cases. The results further indicate that in \ac{sm} \ac{mimo} systems, using a \textit{tuned box-constrained} \ac{lasso} detector achieves a close-to-optimal error rate with a tractable computational complexity. %Further discussions on applications of the results are skipped here and left as a potential future work.

%In this respect, a generic \ac{rls}-based recovery algorithm has been characterized in the large-system limit. Using the asymptotic results, the optimal error rate for detection of spatially modulated signals in uplink of multiuser massive \ac{mimo} systems has been derived. The results have been further used to tune the class of \ac{lasso}-type detectors which are computationally efficient. The scope of applications does not limit to the particular ones investigated in this paper. Following the generality of the results, one could further employ them to address various other applications.

\appendices

\section{Proof of Theorem~\ref{th1}}
\label{app:A}
From Stirling's formula we know that for any integer $M$
\begin{align}
\sqrt{2 \pi} \exp\set{-M} \leq \frac{M!}{M^{M+0.5}} \leq \exp\set{-M+1}.
\end{align}
Hence, we can write
\begin{align}
\frac{\sqrt{2\pi}}{e^2} \left. \Theta_0  \right. \Theta_1
\leq {M_{\rm u} \choose L_{\rm u} } \leq 
\frac{e}{2\pi} \left. \Theta_0 \right. \Theta_1
\end{align}
where we define $\Theta_0$ and $\Theta_1$ as
\begin{subequations}
\begin{align}
\Theta_0 &\coloneqq \sqrt{\frac{M_{\rm u}}{ L_{\rm u}\brc{M_{\rm u}-L_{\rm u}} }}\\
\Theta_1 &\coloneqq \frac{M_{\rm u}^{M_{\rm u}}}{ L_{\rm u}^{L_{\rm u}} \brc{M_{\rm u}-L_{\rm u}}^{M_{\rm u}-L_{\rm u}} }.
\end{align}
\end{subequations}
From the definition of $R_{\rm u} = I + \left. L_{\rm u} \right. S$, we can conclude that
\begin{align}
\log \frac{\sqrt{2\pi}}{e^2} - 1 + \Xi
\stackrel{\star}{<} R_{\rm u} \leq
\log \frac{e}{2\pi} + \Xi \label{eq:R_u}
\end{align}
where $\Xi$ is given by
\begin{align}
\Xi = \left. \log \Theta_0  \right. +\log\Theta_1 + \left. L_{\rm u} \right. S
\end{align}
and $\star$ follows the fact that 
\begin{align}
I = \left\lfloor \log {M_{\rm u} \choose L_{\rm u} } \right\rfloor > \log {M_{\rm u} \choose L_{\rm u} } -1.
\end{align}
Considering the definition of the activity ratio, $\log\Theta_0$ is given by
\begin{align}
\log \Theta_0 = - \frac{1}{2} \log M_{\rm u} - \frac{1}{2} \left[ \log \eta + \log \brc{1-\eta} \right].
\end{align}
Moreover, for the term $\log \Theta_1$, we have
\begin{subequations}
\begin{align}
\log \Theta_1 &= M_{\rm u}\log \frac{M_{\rm u}}{M_{\rm u} - L_{\rm u}} - L_{\rm u} \log \frac{L_{\rm u}}{M_{\rm u} - L_{\rm u}} \\
&= -M_{\rm u} \left. \eta \right. \log \eta -M_{\rm u} \brc{ 1-\eta } \log \brc{ 1-\eta } \\
&= M_{\rm u} H_2 \brc{\eta}.
\end{align}
\end{subequations}
Substituting into \eqref{eq:R_u}, and noting that $\bar{R} = R_{\rm u} / M_{\rm u}$, we have
\begin{align}
I_{\rm d} < \bar{R} \leq I_{\rm u} \label{eq:R_u2}
\end{align}
with $I_{\rm d}$ and $I_{\rm u}$ defined in Theorem~\ref{th1}.
%where we define
%\begin{subequations}
%\begin{align}
%C_{\rm d} &= 2\log \frac{\sqrt{2\pi}}{e^2}-2- \log \brc{\eta-\eta^2}\\
%C_{\rm u} &= 2\log \frac{e}{2\pi}- \log \brc{\eta-\eta^2}.
%\end{align}
%\end{subequations}
Defining the function 
\begin{align}
f_{\rm fix}\brc{x} \coloneqq \bar{R} - H_2 \brc{\eta} - \left. \eta \right. S + \frac{\log M_{\rm u} - x}{2M_{\rm u} }
\end{align}
for a given $M_{\rm u}$ and $\eta$, we conclude from \eqref{eq:R_u2} that $f_{\rm fix}\brc{C_{\rm d}} \geq 0$ and $f_{\rm fix}\brc{C_{\rm u}} < 0$. The method of intervals implies that there exist a constant $C\in \left( C_{\rm d} , C_{\rm u} \right]$, such that $f_{\rm fix}\brc{C} = 0$. This concludes the proof.

\section{Proof of Theorem~\ref{proposition:RS}}
\label{app:B}
The proof of Theorem~\ref{proposition:RS} follows Proposition~1 in \cite{bereyhi2018MAP}. To start with the proof, let us define $\bi = \left[ i_1 , \ldots, i_K\right]^\trp$, with $i_k$ being the modulation index of user $k$. For a given realization of $\bi$, the average distortion is defined as
\begin{subequations}
\begin{align}
D\brc{\bi} &\coloneqq \lim_{M\uparrow \infty} \frac{1}{M} \sum_{m=1}^M \Ex{ F_{\rm D} \brc{x^\star_m;x_m} \vert \bi }{}\\
&= \lim_{M\uparrow \infty} \frac{1}{M} \sum_{m=1}^M \int F_{\rm D} \brc{x^\star_m;x_m} \dif P\brc{ x^\star_m,x_m \vert \bi }
\end{align}
\end{subequations}
with $P\brc{ x^\star_m,x_m \vert \bi }$ denoting the cumulative joint distribution of $\brc{x^\star_m,x_m}$ conditioned on realization $\bi$. The joint distribution of the transmit symbols conditioned on $\bi$ is given by
\begin{subequations}
\begin{align}
\rmp\brc{\bx|\bi} &= \rmp\brc{\bx_1 , \ldots ,\bx_K|\bi} \\
&= \left. \prod_{k=1}^K \right.  \prod_{\ell\left.\in\right.\setL\brc{i_k}} \rmp_{s}\brc{x_{k,\ell} } \prod_{\ell\left.\notin\right.\setL\brc{i_k}} \mone\set{ x_{k,\ell} = 0}
\end{align}
\end{subequations}
where $x_{k,\ell}$ denotes the $\ell$-th entry of $\bx_k$ and $\setL\brc{i_k}$ is the activity support of user $k$. By defining
\begin{align}
\Sp\bi = \set{ m \in [M] : x_m \neq 0 },
\end{align}
one can write the conditional distribution as
\begin{align}
\rmp\brc{\bx|\bi} &= \left. \prod_{ m \in \Sp{\bi} } \right. \rmp_{s}\brc{x_{m} }.
\end{align}
One can view $\rmp\brc{\bx|\bi}$ as a vector of two \ac{iid} blocks. The entries of the first block are those belonging to $\Sp\bi$. This block is distributed \ac{iid} with $\rmp_{s}\brc{x_m}$. The second block consists of the remaining entries and is \ac{iid} with $\delta\brc{x_m}$.

For an asymmetric transmit signal with \ac{iid} blocks, the asymptotic performance of a \ac{map} estimator is characterized in \cite{bereyhi2018MAP}. Using Proposition~1 in \cite{bereyhi2018MAP}, we conclude that under the given assumptions\footnote{These assumptions include replica continuity and replica symmetry. For details, see \cite[Section 5]{bereyhi2018MAP}.}
\begin{align}
D \brc{\bi} \coloneqq \lim_{M\uparrow \infty} \frac{1}{M} \sum_{m=1}^M \Ex{ F_{\rm D} \brc{\xx^\star_m;x_m} \vert \bi}{}
\end{align}
where $\xx^\star_m$ is defined for a complex zero-mean unit-variance Gaussian $\zz$ as
\begin{align}
\xx^\star_m \coloneqq \rls{\setX}{x_m + \theta\brc{c^\star,q^\star} \zz \vert \tau \brc{c^\star} }. \label{ScalarRLS}
\end{align}
$\theta\brc{c,q}$ and $\tau\brc{c}$ are defined by \eqref{eq:rs-4a} and \eqref{eq:rs-4b} in Definition~\ref{def:Dec}, respectively, and $c^\star$ and $q^\star$ are a pair of solutions to the fixed-point equations
\begin{subequations}
	\begin{align}
	\frac{ \theta\brc{c,q} }{\tau\brc{c} } c &= \lim_{M\uparrow \infty}  \frac{1}{M} \sum_{m=1}^M \Ex{ \real{\brc{\xx^\star_m - x_m } \zz^*} \vert \bi }{} \label{Fix_local_1} \\
	q &= \lim_{M\uparrow \infty} \frac{1}{M} \sum_{m=1}^M \Ex{ \abs{\xx^\star_m - x_m}^2  \vert \bi }{}. \label{Fix_local_2}
	\end{align}
\end{subequations}

Following the structure of $\rmp\brc{\bx|\bi} $, we conclude that 
\begin{align}
D \brc{\bi} = \lim_{M\uparrow \infty} & \left[ \frac{\abs{ \Sp{\bi} } }{M} \Ex{ F_{\rm D} \brc{\mathrm{s}^\star; \mathrm{s} } \vert \bi}{} + \brc{1-\frac{\abs{ \Sp{\bi} } }{M}} \Ex{ F_{\rm D} \brc{\mathrm{s}_0^\star; 0 } \vert \bi}{} \right] \label{EQ1}
\end{align}
where $\mathrm{s}$ is a random variable distributed with $\rmp_s\brc{\cdot}$, and $\mathrm{s}^\star$ and $\mathrm{s}_0^\star$ are defined as in \eqref{ScalarRLS} by setting $x_m = \mathrm{s}$ and $x_m=0$, respectively. Noting that $\abs{ \Sp{\bi} } / {M} = \eta$, and that the expectation does not depend on the realization of $\bi$, we have
\begin{align}
D \brc{\bi} = \eta \Ex{ F_{\rm D} \brc{\mathrm{s}^\star; \mathrm{s} } }{} + \brc{1-\eta } \Ex{ F_{\rm D} \brc{\mathrm{s}_0^\star; 0 } }{}.
\end{align}
By defining $\xx = \psi \mathrm{s}$ for the Bernoulli random variable $\psi$ of Definition~\ref{def:Dec}, $D \brc{\bi}$ is given by
\begin{align}
D \brc{\bi} = \Ex{ F_{\rm D} \brc{\xx^\star\brc{c^\star , q^\star}; \xx } }{} \label{EQ2}
\end{align}
for $\xx^\star\brc{c , q}$ given in Definition~\ref{def:Dec}. By some lines of derivations analogous to those given in \eqref{EQ1}-\eqref{EQ2}, it is shown that the fixed-point equations in \eqref{Fix_local_1} and \eqref{Fix_local_2} reduce to those given in Theorem~\ref{proposition:RS}.

Finally, by noting that the asymptotic expression for $D\brc{\bi}$ does not depend on $\bi$, we can write
\begin{subequations}
	\begin{align}
	D &= \Ex{D\brc{\bi}  }{\bi}\\
	&= D_{\rm asy} \brc{c^\star,q^\star}.
	\end{align}
\end{subequations}
This concludes the proof.

\bibliography{ref}

\begin{thebibliography}{10}
\providecommand{\url}[1]{#1}
\csname url@samestyle\endcsname
\providecommand{\newblock}{\relax}
\providecommand{\bibinfo}[2]{#2}
\providecommand{\BIBentrySTDinterwordspacing}{\spaceskip=0pt\relax}
\providecommand{\BIBentryALTinterwordstretchfactor}{4}
\providecommand{\BIBentryALTinterwordspacing}{\spaceskip=\fontdimen2\font plus
\BIBentryALTinterwordstretchfactor\fontdimen3\font minus
  \fontdimen4\font\relax}
\providecommand{\BIBforeignlanguage}[2]{{%
\expandafter\ifx\csname l@#1\endcsname\relax
\typeout{** WARNING: IEEEtran.bst: No hyphenation pattern has been}%
\typeout{** loaded for the language `#1'. Using the pattern for}%
\typeout{** the default language instead.}%
\else
\language=\csname l@#1\endcsname
\fi
#2}}
\providecommand{\BIBdecl}{\relax}
\BIBdecl

\bibitem{bereyhi2019rls}
A.~Bereyhi, S.~Asaad, B.~G{\"a}de, and R.~R. M{\"u}ller, ``{RLS}-based
  detection for massive spatial modulation {MIMO},'' \emph{Proc. IEEE
  International Symposium on Information Theory (ISIT)}, pp. 1167--1171, July
  2019, Paris, France.

\bibitem{mesleh2006spatial}
R.~Mesleh, H.~Haas, C.~W. Ahn, and S.~Yun, ``Spatial modulation-a new low
  complexity spectral efficiency enhancing technique,'' \emph{Proc. First
  International Conference on Communications and Networking}, pp. 1--5, October
  2006, China.

\bibitem{mesleh2008spatial}
R.~Mesleh, H.~Haas, S.~Sinanovic, C.~W. Ahn, and S.~Yun, ``Spatial
  modulation,'' \emph{IEEE Transactions on Vehicular Technology}, vol.~57,
  no.~4, p. 2228, July 2008.

\bibitem{jeganathan2008spatial}
J.~Jeganathan, A.~Ghrayeb, and L.~Szczecinski, ``Spatial modulation: Optimal
  detection and performance analysis,'' \emph{IEEE Communications Letters},
  vol.~12, no.~8, August 2008.

\bibitem{di2011spatial}
M.~Di~Renzo, H.~Haas, and P.~M. Grant, ``Spatial modulation for
  multiple-antenna wireless systems: A survey,'' \emph{IEEE Communications
  Magazine}, vol.~49, no.~12, December 2011.

\bibitem{yang2014design}
P.~Yang, M.~Di~Renzo, Y.~Xiao, S.~Li, and L.~Hanzo, ``Design guidelines for
  spatial modulation,'' \emph{IEEE Communications Surveys \& Tutorials},
  vol.~17, no.~1, pp. 6--26, May 2014.

\bibitem{jeganathan2008generalized}
J.~Jeganathan, A.~Ghrayeb, and L.~Szczecinski, ``Generalized space shift keying
  modulation for {MIMO} channels,'' \emph{Proc. 19th International Symposium on
  Personal, Indoor and Mobile Radio Communications (PIMRC)}, pp. 1--5,
  September 2008.

\bibitem{jeganathan2009space}
J.~Jeganathan, A.~Ghrayeb, L.~Szczecinski, and A.~Ceron, ``Space shift keying
  modulation for {MIMO} channels,'' \emph{IEEE Transactions on Wireless
  Communications}, vol.~8, no.~7, pp. 3692--3703, July 2009.

\bibitem{di2010general}
M.~Di~Renzo and H.~Haas, ``A general framework for performance analysis of
  space shift keying {(SSK)} modulation for {MISO} correlated {Nakagami-m}
  fading channels,'' \emph{IEEE Transactions on Communications}, vol.~58,
  no.~9, pp. 2590--2603, September 2010.

\bibitem{Mseleh2011SSK}
R.~{Mesleh}, S.~{Ikki}, and M.~{Alwakeel}, ``Performance analysis of space
  shift keying with amplify and forward relaying,'' \emph{IEEE Communications
  Letters}, vol.~15, no.~12, pp. 1350--1352, October 2011.

\bibitem{di2011bit}
M.~Di~Renzo and H.~Haas, ``Bit error probability of space-shift keying {MIMO}
  over multiple-access independent fading channels,'' \emph{IEEE Transactions
  on Vehicular Technology}, vol.~60, no.~8, pp. 3694--3711, September 2011.

\bibitem{di2012space}
M.~Di~Renzo, D.~De~Leonardis, F.~Graziosi, and H.~Haas, ``Space shift keying
  {(SSK) MIMO} with practical channel estimates,'' \emph{IEEE Transactions on
  Communications}, vol.~60, no.~4, pp. 998--1012, February 2012.

\bibitem{popoola2013error}
W.~O. Popoola, E.~Poves, and H.~Haas, ``Error performance of generalised space
  shift keying for indoor visible light communications,'' \emph{IEEE
  Transactions on Communications}, vol.~61, no.~5, pp. 1968--1976, March 2013.

\bibitem{basar2011space}
E.~{Başar}, U.~{Aygölü}, E.~{Panayirci}, and H.~V. {Poor}, ``Space-time
  block coded spatial modulation,'' \emph{IEEE Transactions on Communications},
  vol.~59, no.~3, pp. 823--832, March 2011.

\bibitem{bian2015differential}
Y.~{Bian}, X.~{Cheng}, M.~{Wen}, L.~{Yang}, H.~V. {Poor}, and B.~{Jiao},
  ``Differential spatial modulation,'' \emph{IEEE Transactions on Vehicular
  Technology}, vol.~64, no.~7, pp. 3262--3268, July 2015.

\bibitem{fu2010generalised}
J.~Fu, C.~Hou, W.~Xiang, L.~Yan, and Y.~Hou, ``Generalised spatial modulation
  with multiple active transmit antennas,'' \emph{Proc. IEEE Global
  Communications Conference (GLOBECOM) Workshops}, pp. 839--844, December 2010,
  USA.

\bibitem{wang2012generalised}
J.~Wang, S.~Jia, and J.~Song, ``Generalised spatial modulation system with
  multiple active transmit antennas and low complexity detection scheme,''
  \emph{IEEE Transactions on Wireless Communications}, vol.~11, no.~4, pp.
  1605--1615, April 2012.

\bibitem{cheng2015enhanced}
C.-C. Cheng, H.~Sari, S.~Sezginer, and Y.~T. Su, ``Enhanced spatial modulation
  with multiple signal constellations,'' \emph{IEEE Transactions on
  Communications}, vol.~63, no.~6, pp. 2237--2248, June 2015.

\bibitem{hoydis2013massive}
J.~Hoydis, S.~Ten~Brink, and M.~Debbah, ``Massive {MIMO} in the {UL/DL} of
  cellular networks: How many antennas do we need?'' \emph{IEEE Journal on
  Selected Areas in Communications}, vol.~31, no.~2, pp. 160--171, January
  2013.

\bibitem{larsson2014massive}
E.~G. Larsson, O.~Edfors, F.~Tufvesson, and T.~L. Marzetta, ``Massive {MIMO}
  for next generation wireless systems,'' \emph{IEEE Communications Magazine},
  vol.~52, no.~2, pp. 186--195, February 2014.

\bibitem{bjornson2015massive}
E.~Bj{\"o}rnson, M.~Matthaiou, and M.~Debbah, ``Massive {MIMO} with non-ideal
  arbitrary arrays: Hardware scaling laws and circuit-aware design,''
  \emph{IEEE Transactions on Wireless Communications}, vol.~14, no.~8, pp.
  4353--4368, April 2015.

\bibitem{bjornson2016massive}
E.~Bj{\"o}rnson, E.~G. Larsson, and T.~L. Marzetta, ``Massive {MIMO}: Ten myths
  and one critical question,'' \emph{IEEE Communications Magazine}, vol.~54,
  no.~2, pp. 114--123, February 2016.

\bibitem{gao2018low}
X.~Gao, L.~Dai, and A.~M. Sayeed, ``Low {RF}-complexity technologies to enable
  millimeter-wave {MIMO} with large antenna array for {5G} wireless
  communications,'' \emph{IEEE Communications Magazine}, vol.~56, no.~4, pp.
  211--217, February 2018.

\bibitem{kuehne2018analog}
T.~K{\"u}hne and G.~Caire, ``An analog module for hybrid massive {MIMO}
  testbeds demonstrating beam alignment algorithms,'' \emph{Proc. 22nd
  International ITG Workshop on Smart Antennas (WSA)}, pp. 1--8, March 2018,
  Bochum, Germany.

\bibitem{kuehne2020performance}
T.~K{\"u}hne, X.~Song, G.~Caire, K.~Rasilainen, T.~H. Le, M.~Rossi, I.~Ndip,
  and C.~Fager, ``Performance simulation of a {5G} hybrid beamforming
  millimeter-wave system,'' \emph{Proc. 24th International ITG Workshop on
  Smart Antennas (WSA)}, pp. 1--6, February 2020, Hamburg, Germany.

\bibitem{alkhateeb2014channel}
A.~Alkhateeb, O.~El~Ayach, G.~Leus, and R.~W. Heath, ``Channel estimation and
  hybrid precoding for millimeter wave cellular systems,'' \emph{IEEE Journal
  of Selected Topics in Signal Processing}, vol.~8, no.~5, pp. 831--846, July
  2014.

\bibitem{gao2016energy}
X.~Gao, L.~Dai, S.~Han, I.~Chih-Lin, and R.~W. Heath, ``Energy-efficient hybrid
  analog and digital precoding for {mmWave} {MIMO} systems with large antenna
  arrays,'' \emph{IEEE Journal on Selected Areas in Communications}, vol.~34,
  no.~4, pp. 998--1009, March 2016.

\bibitem{mendez2016hybrid}
R.~Mendez-Rial, C.~Rusu, N.~Gonz{\'a}lez-Prelcic, A.~Alkhateeb, and R.~W.
  Heath, ``Hybrid {MIMO} architectures for millimeter wave communications:
  Phase shifters or switches?'' \emph{IEEE Access}, vol.~4, pp. 247--267,
  January 2016.

\bibitem{asaad2017asymptotic}
S.~Asaad, A.~Bereyhi, M.~A. Sedaghat, R.~M\"uller, and A.~M. Rabiei,
  ``Asymptotic performance analysis of spatially reconfigurable antenna
  arrays,'' \emph{Proc. 21st International ITG Workshop on Smart Antennas
  (WSA)}, pp. 1--6, March 2017, Berlin, Germany.

\bibitem{bereyhi2019papr}
A.~Bereyhi, V.~Jamali, R.~R. M{\"u}ller, G.~Fischer, R.~Schober, and A.~M.
  Tulino, ``{PAPR}-limited precoding in massive {MIMO} systems with reflect-and
  transmit-array antennas,'' \emph{IEEE 53rd Asilomar Conference on Signals,
  Systems, and Computers}, pp. 1690--1694, November 2019, USA.

\bibitem{bereyhi2017wsa}
A.~Bereyhi, M.~A. Sedaghat, S.~Asaad, and R.~M\"uller, ``Nonlinear precoders
  for massive {MIMO} systems with general constraints,'' \emph{Proc. 21st
  International ITG Workshop on Smart Antennas (WSA)}, March 2017, Berlin,
  Germany.

\bibitem{sedaghat2017least}
M.~A. Sedaghat, A.~Bereyhi, and R.~R. M{\"u}ller, ``Least square error
  precoders for massive {MIMO} with signal constraints: Fundamental limits,''
  \emph{IEEE Transactions on Wireless Communications}, vol.~17, no.~1, pp.
  667--679, November 2017.

\bibitem{bereyhi2019glse}
A.~Bereyhi, M.~A. Sedaghat, R.~R. M{\"u}ller, and G.~Fischer, ``{GLSE}
  precoders for massive {MIMO} systems: Analysis and applications,'' \emph{IEEE
  Transactions on Wireless Communications}, vol.~18, no.~9, pp. 4450--4465,
  2019.

\bibitem{li2014energy}
H.~Li, L.~Song, and M.~Debbah, ``Energy efficiency of large-scale multiple
  antenna systems with transmit antenna selection,'' \emph{IEEE Transactions on
  Communications}, vol.~62, no.~2, pp. 638--647, 2014.

\bibitem{gao2015massive}
X.~Gao, O.~Edfors, F.~Tufvesson, and E.~G. Larsson, ``Massive {MIMO} in real
  propagation environments: Do all antennas contribute equally?'' \emph{IEEE
  Transactions on Communications}, vol.~63, no.~11, pp. 3917--3928, July 2015.

\bibitem{liu2016efficient}
X.~Liu and X.~Wang, ``Efficient antenna selection and user scheduling in {5G}
  massive {MIMO-NOMA} system,'' \emph{IEEE 83rd Vehicular Technology Conference
  (VTC Spring)}, pp. 1--5, July 2016, Nanjing, China.

\bibitem{asaad2017tas}
S.~Asaad, A.~Bereyhi, R.~R. M{\"u}ller, and A.~M. Rabiei, ``Asymptotics of
  transmit antenna selection: Impact of multiple receive antennas,''
  \emph{Proc. IEEE International Conference on Communications (ICC)}, pp. 1--6,
  May 2017, Paris, France.

\bibitem{bereyhi2017isit}
A.~Bereyhi, M.~A. Sedaghat, and R.~M\"uller, ``Asymptotics of nonlinear {LSE}
  precoders with applications to transmit antenna selection,'' \emph{Proc. IEEE
  International Symposium on Information Theory (ISIT)}, July 2017, Aachen,
  Germany.

\bibitem{asaad2018massive}
S.~Asaad, A.~M. Rabiei, and R.~R. M{\"u}ller, ``Massive {MIMO} with antenna
  selection: Fundamental limits and applications,'' \emph{IEEE Transactions on
  Wireless Communications}, vol.~17, no.~12, pp. 8502--8516, November 2018.

\bibitem{asaad2018optimal}
S.~Asaad, A.~Bereyhi, A.~M. Rabiei, R.~R. M{\"u}ller, and R.~F. Schaefer,
  ``Optimal transmit antenna selection for massive {MIMO} wiretap channels,''
  \emph{IEEE Journal on Selected Areas in Communications}, pp. 817--828, April
  2018.

\bibitem{bereyhi2018stepwise}
A.~Bereyhi, S.~Asaad, and R.~R. M{\"u}ller, ``Stepwise transmit antenna
  selection in downlink massive multiuser {MIMO},'' \emph{Proc. 22nd
  International ITG Workshop on Smart Antennas (WSA)}, pp. 1--8, March 2018,
  Bochum, Germany.

\bibitem{sedaghat2014novel}
M.~A. Sedaghat, R.~R. M{\"u}ller, and G.~Fischer, ``A novel single-{RF}
  transmitter for massive {MIMO},'' \emph{Proc. 18th International ITG Workshop
  on Smart Antennas (WSA)}, pp. 1--8, March 2014, Erlangen, Germany.

\bibitem{sedaghat2016load}
M.~A. Sedaghat, V.~I. Barousis, R.~R. M{\"u}ller, and C.~B. Papadias, ``Load
  modulated arrays: a low-complexity antenna,'' \emph{IEEE Communications
  Magazine}, vol.~54, no.~3, pp. 46--52, March 2016.

\bibitem{gade2017novel}
B.~G{\"a}de, M.~A. Sedaghat, C.~Rachinger, R.~M{\"u}ller, and G.~Fischer, ``A
  novel single-{RF} outphasing {MIMO} architecture,'' \emph{Proc. 21st
  International ITG Workshop on Smart Antennas (WSA)}, pp. 1--7, March 2017,
  Berlin, Germany.

\bibitem{bereyhi2020single}
A.~Bereyhi, V.~Jamali, R.~R. M{\"u}ller, A.~M. Tulino, G.~Fischer, and
  R.~Schober, ``A single-{RF} architecture for multiuser massive {MIMO} via
  reflecting surfaces,'' \emph{Proc. IEEE International Conference on
  Acoustics, Speech and Signal Processing (ICASSP)}, pp. 8688--8692, May 2020,
  Barcelona, Spain.

\bibitem{gade2019fair}
B.~G{\"a}de, A.~Bereyhi, S.~Asaad, and R.~R. M{\"u}ller, ``A fair comparison
  between spatial modulation and antenna selection in massive {MIMO} systems,''
  \emph{Proc. 23rd International ITG Workshop on Smart Antennas (WSA)}, pp.
  1--6, April 2019, Vienna, Austria.

\bibitem{wen2015low}
M.~{Wen}, X.~{Cheng}, Y.~{Bian}, and H.~V. {Poor}, ``A low-complexity near-{ML}
  differential spatial modulation detector,'' \emph{IEEE Signal Processing
  Letters}, vol.~22, no.~11, pp. 1834--1838, November 2015.

\bibitem{candes2006robust}
E.~J. Cand{\`e}s, J.~Romberg, and T.~Tao, ``Robust uncertainty principles:
  Exact signal reconstruction from highly incomplete frequency information,''
  \emph{IEEE Transactions on Information Theory}, vol.~52, no.~2, pp. 489--509,
  January 2006.

\bibitem{donoho2006compressed}
D.~L. Donoho, ``Compressed sensing,'' \emph{IEEE Transactions on Information
  Theory}, vol.~52, no.~4, pp. 1289--1306, April 2006.

\bibitem{yu2012compressed}
C.-M. Yu, S.-H. Hsieh, H.-W. Liang, C.-S. Lu, W.-H. Chung, S.-Y. Kuo, and S.-C.
  Pei, ``Compressed sensing detector design for space shift keying in {MIMO}
  systems,'' \emph{IEEE Communications Letters}, vol.~16, no.~10, pp.
  1556--1559, September 2012.

\bibitem{garcia2015low}
A.~Garcia-Rodriguez and C.~Masouros, ``Low-complexity compressive sensing
  detection for spatial modulation in large-scale multiple access channels,''
  \emph{IEEE Transactions on Communications}, vol.~63, no.~7, pp. 2565--2579,
  May 2015.

\bibitem{xiao2017compressed}
L.~Xiao, Y.~Xiao, C.~Xu, X.~Lei, P.~Yang, S.~Li, and L.~Hanzo,
  ``Compressed-sensing assisted spatial multiplexing aided spatial
  modulation,'' \emph{IEEE Transactions on Wireless Communications}, vol.~17,
  no.~2, pp. 794--807, November 2017.

\bibitem{hemadeh2018compressed}
I.~A. Hemadeh, S.~Lu, M.~El-Hajjar, and L.~Hanzo, ``Compressed sensing-aided
  index modulation improves space-time shift keying assisted millimeter-wave
  communications,'' \emph{IEEE Access}, vol.~6, pp. 64\,742--64\,756, October
  2018.

\bibitem{atitallah2017box}
I.~B. Atitallah, C.~Thrampoulidis, A.~Kammoun, T.~Y. Al-Naffouri, M.-S.
  Alouini, and B.~Hassibi, ``The {BOX-LASSO} with application to {GSSK}
  modulation in massive {MIMO} systems,'' \emph{Proc. IEEE International
  Symposium on Information Theory (ISIT)}, pp. 1082--1086, July 2017, Aachen,
  Germany.

\bibitem{thrampoulidis2018symbol}
C.~Thrampoulidis, W.~Xu, and B.~Hassibi, ``Symbol error rate performance of
  box-relaxation decoders in massive {MIMO},'' \emph{IEEE Transactions on
  Signal Processing}, vol.~66, no.~13, pp. 3377--3392, April 2018.

\bibitem{alrashdi2019precise}
A.~M. Alrashdi, I.~B. Atitallah, and T.~Y. Al-Naffouri, ``Precise performance
  analysis of the box-elastic net under matrix uncertainties,'' \emph{IEEE
  Signal Processing Letters}, vol.~26, no.~5, pp. 655--659, February 2019.

\bibitem{alrashdi2020box}
A.~M. Alrashdi, H.~Sifaou, A.~Kammoun, M.-S. Alouini, and T.~Y. Al-Naffouri,
  ``Box-relaxation for {BPSK} recovery in massive {MIMO}: A precise analysis
  under correlated channels,'' \emph{in Proc. IEEE International Conference on
  Communications (ICC)}, pp. 1--6, 2020.

\bibitem{bereyhi2018MAP}
A.~Bereyhi and R.~R. Müller, ``Maximum-a-posteriori signal recovery with prior
  information: Applications to compressive sensing,'' \emph{Proc. IEEE
  International Conference on Acoustics, Speech and Signal Processing
  (ICASSP)}, pp. 4494--4498, April 2018, Calgary, Alberta, Canada.

\bibitem{cakmak2018capacity}
B.~{Çakmak}, R.~R. {Müller}, and B.~H. {Fleury}, ``Capacity scaling in {MIMO}
  systems with general unitarily invariant random matrices,'' \emph{IEEE
  Transactions on Information Theory}, vol.~64, no.~5, pp. 3825--3841, May
  2018.

\bibitem{bereyhi2016statistical}
A.~Bereyhi, R.~R. M{\"u}ller, and H.~Schulz-Baldes, ``Statistical mechanics of
  map estimation: General replica ansatz,'' \emph{IEEE Transactions on
  Information Theory}, vol.~65, no.~12, pp. 7896--7934, August 2019.

\bibitem{akhiezer1965classical}
N.~I. Akhiezer, \emph{The Classical Moment Problem and Some Related Questions
  in Analysis}.\hskip 1em plus 0.5em minus 0.4em\relax Oliver \& Boyd, 1965,
  vol.~5.

\bibitem{tulino2004random}
A.~M. Tulino and S.~Verd{\'u}, \emph{Random Matrix Theory and Wireless
  Communications}.\hskip 1em plus 0.5em minus 0.4em\relax Foundations and
  Trends in Communications and Information Theory, Now Publishers Inc, 2004,
  USA.

\bibitem{korada2011applications}
S.~B. Korada and A.~Montanari, ``Applications of the lindeberg principle in
  communications and statistical learning,'' \emph{IEEE Transactions on
  Information Theory}, vol.~57, no.~4, pp. 2440--2450, March 2011.

\bibitem{oymak2017universality}
S.~Oymak and J.~A. Tropp, ``Universality laws for randomized dimension
  reduction, with applications,'' \emph{Information and Inference: A Journal of
  the IMA}, vol.~7, no.~3, pp. 337--446, November 2017.

\end{thebibliography}
\bibliographystyle{IEEEtran}
\end{document}